\documentclass{CSML}

\def\dOi{11(4:1)2015}
\lmcsheading%
{\dOi}
{1--39}
{}
{}
{Apr.~\phantom08, 2014}
{Oct.~\phantom05, 2015}
{}

\ACMCCS{[{\bf Information systems}]: Data management
  systems---Database design and models---Graph-based database models;
  Data management systems---Query languages---Query languages for
  non-relational engines; [{\bf Theory of computation}]: Theory and
  algorithms for application domains---Database theory---Database
  query languages (principles); Theory and algorithms for application
  domains---Database theory---Logic and databases; Formal languages
  and automata theory---Regular languages}


\usepackage{verbatim,hyperref}

\usepackage{amssymb,amsfonts}
\usepackage{amsmath,tikz} 
\usepackage{amssymb,bm,latexsym}
\usepackage{stmaryrd}

\usepackage{amsmath,amsfonts,amssymb}

\newcommand{\V}{{\mathcal D}} 
\newcommand{\boxtheorem}{\hfill $\Box$}
\newcommand{\pspace}{{\sc Pspace}}

\newcommand{\nlogspace}{{\sc NLogspace}}
\newcommand{\expspace}{{\sc Expspace}}  
\newcommand{\M}{{\mathcal M}} 
\newcommand{\FO}{{\rm FO}} 
\newcommand{\A}{{\mathcal A}} 
\newcommand{\OMIT}[1]{}
\renewcommand{\l}{\ell} 

\newcommand{\D}{{\mathcal D}} 

\newcommand{\T}{{\mathcal T}}

\renewcommand{\S}{{\mathcal S}}

\newcommand{\N}{\mathbb{N}}

\newtheorem*{proposition*}{Proposition}
\newtheorem*{notation*}{Notation}
\newtheorem*{theorem*}{Theorem}
\newtheorem*{definition*}{Definition}
\newtheorem*{lemma*}{Lemma}

\newcommand{\wl}{ \mathit{WL}}
\newcommand{\rl}{ \mathit{RL}}

\newcommand{\tow}[2]{\mathit{tower}(#1,#2)}

\newcommand{\sem}[1]{\llbracket #1 \rrbracket}

\newtheorem{theorem}{Theorem}[section]
\newtheorem{lemma}[theorem]{Lemma}
\newtheorem{corollary}[theorem]{Corollary}
\newtheorem{definition}{Definition}[section]
\newtheorem{proposition}[theorem]{Proposition}
\newtheorem{claim}{Claim}[theorem]

\newtheorem{example}{Example}[section]

\begin{document}

\title[Expressive Path Queries on Graphs with Data]
{Expressive Path Queries on Graphs with Data\rsuper*}

\author[P.~Barcel{\'o}]{Pablo Barcel{\'o}\rsuper a}	
\address{{\lsuper{a,b}}Center for Semantic Web Research \& Department of Computer Science, 
University of Chile}	
\email{\{pbarcelo, gaelle\}@dcc.uchile.cl}  

\author[G.~Fontaine]{Gaelle Fontaine\rsuper b}	
\address{\vspace{-18 pt}}	

\author[A. W-Lin]{Anthony Widjaja Lin\rsuper c}	
\address{{\lsuper c}Yale-NUS College, Singapore}	
\email{anthony.w.to@gmail.com}  

\thanks{{\lsuper{a,c}}Barcel\'o is funded by the Millennium Nucleus Center
for Semantic Web Research under Grant NC120004 and Fontaine by Fondecyt postdoctoral 
grant 3130491. Part of this work done when Lin visited Barcel\'o funded by 
Fondecyt grant 1130104. This work was also partially done when Lin was 
at Oxford University supported by EPSRC (H026878).}	

\titlecomment{{\lsuper*}This is the full version of the conference paper \cite{BFW}.}


\begin{abstract}
Graph data models have recently become popular owing to their applications,
e.g., in social networks and the semantic web. Typical navigational query languages 
over graph databases --- such as Conjunctive Regular Path Queries (CRPQs) --- 
cannot express relevant properties of the interaction between the underlying 
data and the topology. Two languages have been recently proposed to overcome 
this problem: {\em walk logic} (WL) and {\em regular expressions with memory} 
(REM). In this paper, we begin by investigating fundamental properties of
WL and REM, i.e., complexity of evaluation problems and expressive power. We
first show that the data complexity of WL is nonelementary, which rules out
its practicality. On the other hand, while REM has low data complexity, we 
point out that many natural data/topology properties of graphs expressible in 
WL cannot be expressed in REM. To this end, we propose 
{\em register logic}, an extension of REM, which we show to be able to express 
many natural graph properties expressible in WL, while at the same time 
preserving the elementariness of data complexity of REMs. It is also 
incomparable to WL in terms of expressive power.
\end{abstract} 

\maketitle

\section{Introduction} 
\label{sec:intro}

Graph databases have gained renewed interest due to 
applications, such as the semantic web, social network
analysis, crime detection networks, software bug detection, biological
networks, and others (e.g., see \cite{AG} for a survey). 
Despite the
importance of querying graph databases, no general agreement has been
reached to date about the kind of features a practical query language
for graph databases should support and about what can be considered 
a reasonable computational cost of query evaluation for the
aforementioned applications. 
  
Typical navigational query languages for graph databases --- including
the conjunctive regular path queries \cite{CMW} and its many
extensions \cite{BLWW} --- suffer from a common drawback: they are
well-suited for expressing relevant properties about the underlying
topology of a graph database, i.e., about the way in which (labeled) 
nodes are connected via (labeled) edges, but not about how such topology  
interacts with the node ids or the data. This drawback is shared by common
specification languages for verification \cite{MC} (e.g. CTL$^*$), which are
evaluated over a similar graph data model (a.k.a. transition systems).
Examples of important queries that combine graph data and topology, but cannot 
be expressed in usual
navigational languages for graph databases, include the following \cite{WL,trial}: (Q1)
{\em Find pairs of people in a social network connected by
  professional links restricted to people of the same age}. 
(Q2) {\em Find pairs of cities $x$ and $y$ in a transportation
  system, such that $y$ can be reached from $x$ using only services
  operated by the same company}. In each one of these queries, the
connectivity between two nodes (i.e., the topology) is constrained by
the data (from an infinite domain, e.g., $\N$), in the sense that we only 
consider 
paths in which all intermediate nodes
satisfy a certain condition (e.g. they are people of the same age).

Two languages, {\em walk logic} and {\em regular expressions with
  memory}, have recently been proposed to overcome this
problem. These languages have different goals:

\smallskip

(a) \underline{Walk logic (WL)} was proposed by Hellings et
al. \cite{WL} as a
  unifying framework for understanding the expressive power of path
  queries over graph databases. Its strength is on the expressiveness side.
  The underlying data model of WL is that of (node or edge)-labeled
  directed graphs. In this context, WL can be seen as a natural
  extension of FO with path quantification,
  plus the ability to check whether positions $p$ and $p'$ in paths
  $\pi$ and $\pi'$, respectively, have the same data values. In their 
  paper, Hellings et
al. assume the restriction that each node carries a distinct data
  value (and, therefore, that this data value serves as an identifier 
for the node). 
However, as we shall see, this makes no difference in terms of the
  results that we can obtain.

  \smallskip  (b) \underline{Regular expressions with memory
    (REMs)} were proposed by Libkin and Vrgo\v{c} \cite{LV} as a
  formalism for comparing data values along a single path,
  while retaining a reasonable complexity for query evaluation. The
  strength of this language is on the side of efficiency.
  The data model of the class of REMs is that of
  edge-labeled directed graphs, in which each node is assigned a data
  value from an infinite domain. REMs define pairs of nodes in the
  graph database that are linked by a path satisfying a given
  condition $c$. Each such condition $c$ is defined in a formalism 
  inspired by the class of {\em register automata} \cite{Kam}, allowing some
  data values to be
  stored in the registers and then compared against other data values.
  The evaluation problem for REMs is \pspace-complete
  (same as for FO over relational databases), and can be solved in
  polynomial time in {\em data complexity} \cite{LV}, i.e.,
  assuming queries to be fixed.\footnote{Recall that data complexity
    is a reasonable measure of complexity in the database scenario
    \cite{Vardi}, since queries are often much smaller than the
    underlying data.} This shows that the language is, in fact,
  well-behaved in terms of the complexity of query evaluation. 

\smallskip

\medskip 

The aim of this paper is to investigate the expressiveness and complexity of 
query evaluation for WL and the class of REMs with the hope of finding 
a navigational query language for data graphs that strikes a good balance 
between these two important aspects of query languages.
\smallskip
%
 %

\noindent
\textbf{Contributions.}
We start by considering WL, which is known to be a powerful formalism
in terms of expressiveness. Little is known about the
cost of query evaluation for this language, save for the decidability of 
the evaluation problem and {\sc NP}-hardness of its data complexity. Our first main 
contribution is to pinpoint the exact complexity of the evaluation problem for 
WL (and thus answering an open problem from \cite{WL}): we prove that it
is non-elementary, and that this holds even in data
complexity, which rules out the practicality of the language. 
\OMIT{
With this not only we solve an open problem from
\cite{WL}, but also completely rule out any applicability of WL as a
practical query language. In addition, our proof
of the non-elementary lower bound uses some very simple queries, which 
severely restricts the
possibility of finding an interesting fragment of the logic with
reasonable data complexity. 
}

We thus move to the class of REMs, which suffers from the opposite
drawback: Although the 
complexity of evaluation for queries in this class is reasonable, 
the expressiveness of the
language is too rudimentary for expressing some important path
properties due to its inability to (i) compare data values in \emph{different} 
paths and (ii) express branching properties of the graph database.
An example of an interesting query that is not expressible as an REM
is the following: (Q) {\em Find pairs of nodes $x$ and $y$, such that there is
a node $z$ and a path $\pi$ from $x$ to $y$ in which each node 
is connected to $z$}. 
Notice that this is 
the query that lies at the basis of the queries (Q1) and (Q2) we presented 
before.  

Our second contribution then is to identify a natural extension of this 
language, called {\em
  register logic} (RL), that closes REMs under Boolean
combinations and existential quantification over nodes, paths and
register assignments. The latter allows the logic to express 
comparisons of data values appearing in different paths, as well as
branching properties of the data. 
This logic is incomparable in expressive power to WL. Besides, many
natural queries relating data and topology in data graphs can be expressed
in RL including: the query (Q), hamiltonicity, the existence of
an Eulerian trail, bipartiteness, and connected graphs with an even number of
nodes.
We then study the complexity of the problem of query evaluation for RL, and
show that it can be solved in elementary time (in particular, that it
is \expspace-complete). This is in contrast to WL, for which even
the data complexity is non-elementary. With respect to data
complexity, we prove that RL is \pspace-complete. 
We then identify a slight extension of its existential-positive fragment,
which  
is tractable (\nlogspace) in data complexity and can express many
queries of interest (including the query (Q)). The idea behind this
extension is that atomic REMs can be enriched with an existential
branching operator -- in the style of the class of {\em nested regular
  expressions} \cite{BLP} -- that increases expressiveness without
affecting the cost of evaluation.

\smallskip
\noindent
{\bf Organization of the paper.} \,   
Section \ref{sec:model} defines our data model. In Section \ref{sec:wl},
we briefly recall the definition of walk logic and some basic results
from \cite{WL}. In Section \ref{sec:nonelem}, we prove that the data complexity
of WL is nonelementary. Section \ref{sec:rl} contains our results concerning
register logic. We conclude in Section \ref{sec:conc} with future work.

\section{The Data Model} 
\label{sec:model}

We start with a definition of our data model: data graphs.
\begin{definition}[Data graph] 
Let $\Sigma$ be a finite alphabet. 
A data graph $G$ over $\Sigma$ is a tuple
$(V,E,\kappa)$, where $V$ is the finite set of nodes, $E \subseteq V
\times \Sigma \times V$ is the set of directed edges labeled in
$\Sigma$ (that is, each triple $(v,a,v') \in E$ is to be seen as
an edge from $v$ to $v'$ in $G$ labeled $a$), and $\kappa : V \to \D$
is a function that assigns a data value in $\D$ to each node in
$V$. 
\end{definition} 
This is the data model adopted by Libkin and Vrgo\v{c} \cite{LV} in their
definition of REMs. In the case of WL \cite{WL}, the authors adopted \emph{graph
databases} as their data model, i.e., data graphs $G = (V,E,\kappa)$ such that
$\kappa$ is injective (i.e. each node carries a different data value). 
In such a case we can think of $\kappa(v)$ as the identifier (id) of
$v$, for each $v \in V$. 
We shall adopt the general model of \cite{LV} since none of 
our complexity results are affected by the data model: upper bounds hold for
data graphs, while all lower bounds are proved in the more restrictive
setting of graph databases. However, for the sake of the comparison
with the expressiveness of WL, many of our examples are constructed in
the scenario of graph databases, that is, when $\kappa(v)$ serves as
an id for node $v$. 

There is also the issue of edge-labeled vs node-labeled data graphs.
Our data model is edge-labeled, but the original one for WL is
node-labeled \cite{WL}. We have chosen to use the former because it is
the standard in the literature \cite{Bar13}. Again,
this choice is inessential, since all the complexity results we
present in the paper remains true if the logics are interpreted
over node-labeled graph databases or data graphs (applying the expected
modifications to the syntax). 


Finally, in several of our examples we use logical formulas to express 
properties of undirected graphs. In each such case we assume that an undirected
graph $H$ is represented as a graph database $G = (V,E,\kappa)$ over unary
alphabet $\Sigma = \{a\}$, where $V$ is the set of nodes of
$H$ and $E$ is a symmetric relation (i.e. $(v,a,v') \in E$ iff $(v',a,v) \in
E$). In particular, since $G = (V,E,\kappa)$ is a graph database we
have that $\kappa$ is injective, i.e., each node is uniquely
determined 
by its data value. 

\section{Walk Logic} 
\label{sec:wl}

WL is an elegant and powerful formalism for defining properties of
paths in graph databases, which was originally proposed in \cite{WL} as a
yardstick for measuring the expressiveness of different path logics. 
 
The syntax of WL is defined with respect to countably infinite sets 
$\Pi$ of {\em path variables} (that we denote as
$\pi,\pi_1,\pi_2, \dots$) and $\T(\pi)$, for each $\pi \in \Pi$,
 of {\em position
  variables} of sort $\pi$. We assume that different sorts are associated with
  distinct position variables.
We denote position variables by
$t,t_1,t_2,\dots$, and write $t^\pi$  
when we need to emphasize that
position variable $t$ is of sort $\pi$. 

\begin{definition}[Walk logic (WL)] 
The set of formulas of WL over finite alphabet $\Sigma$ is defined by
the following grammar, where (i) $a \in \Sigma$, (ii) $t,t_1,t_2$ are position
variables of any sort, (iii) $\pi$ is a path variable, and (iv) 
$t_1^\pi,t_2^\pi$ are position variables of the
same sort $\pi$: 
$$\phi,\phi' \, := \, E_a(t_1^\pi,t_2^\pi)  \, \mid \, t_1^\pi
< t_2^\pi \, \mid \, t_1 \sim t_2 \, \mid \, \neg \phi \, \mid \, \phi
\vee \phi' \, \mid \, \exists t \phi \, \mid \, \exists \pi \phi$$ 
As usual, WL formulas without free variables are called {\em
Boolean}. \qed
\end{definition}  

To define the semantics of WL we need to introduce some 
terminology. 
A {\em path} (a.k.a. \emph{walk} in \cite{WL}) in the data graph 
$G = (V,E,\kappa)$ is a 
finite, nonempty sequence $$\rho \ = \ v_1 a_1 v_2
\cdots v_{n-1} a_{n-1} v_n,$$ 
such that $(v_i,a_i,v_{i+1}) \in E$ for each
$1 \leq i < n$. 
The set of {\em positions} of $\rho$ is $\{1,\dots,n\}$, and $v_i$ is
the node in position $i$ of $\rho$, for $1 \leq i \leq
n$. 
The intuition behind the semantics of  WL formulas is as
follows. Each path variable $\pi$ is interpreted as a path $\rho = v_1
a_1 v_2 \cdots v_{n-1} a_{n-1} v_n$ in the data graph $G$, while each
position variable $t$ of sort $\pi$ is interpreted as a position $1
\leq i \leq n$ in $\rho$ (that is, position variables of sort $\pi$
are interpreted as positions in the path that interprets $\pi$). The
atomic formula $E_a(t_1^\pi,t_2^\pi)$ is true iff $\pi$ is interpreted
as path $\rho = v_1 a_1 v_2 \cdots v_{n-1} a_{n-1} v_n$, the position
$p_2$ that interprets $t_2$ in $\rho$ is the successor of the position
$p_1$ that interprets $t_1$ (i.e. $p_2 = p_1 + 1$), and node in
position $p_1$ is linked in $\rho$ by an $a$-labeled edge to node in
position $p_2$ (that is, $a_{p_1} = a$). In the same way, $t_1^\pi <
t_2^\pi$ holds iff in the path $\rho$ that interprets $\pi$ the
position that interprets $t_1$ is smaller than the one that interprets
$t_2$. Furthermore, $t_1 \sim t_2$ is the case iff the data value
carried by the node in the position assigned to $t_1$ is the same than
the data value carried by the node in the position assigned to $t_2$
(possibly in different paths). We formalize the semantics of WL below.

Let $G = (V,E,\kappa)$ be a data graph and $\phi$ a WL formula. 
Assume that $\S_\phi$ is the set that consists of (i) all 
position variables $t^\pi$ and path
variables $\pi$ such that $t^\pi$ is a free variable of $\phi$, and
(ii) all path variables $\pi$ such that $\pi$ is a free variable of
$\phi$. Intuitively,
$\S_\phi$ defines the set of (both path and position) variables that
are relevant to define the semantics of $\phi$ over $G$.   
An {\em assignment} $\alpha$ for $\phi$ over
$G$ is a mapping that associates a path $\rho = v_1 a_1 v_2 \cdots
v_{n-1} a_{n-1} v_n$ in
$G$ with each path variable $\pi \in \S_\phi$, and a position $1 \leq i \leq
n$ with each position variable of the form $t^\pi$ in $\S_\phi$
(notice that this is well-defined since $\pi \in S_\phi$ every
time a position variable of the form $t^\pi$ is in $S_\phi$). As usual, we denote by
$\alpha[t \to i]$ and $\alpha[\pi \to \rho]$ the assignments that are
equal to $\alpha$ except that $t$ is now assigned position $i$ and
$\pi$ the path $\rho$, respectively. 

We say that {\em $G$ satisfies 
$\phi$ under $\alpha$}, denoted $(G,\alpha) \models \phi$, if one of
the following holds (we omit Boolean combinations which are standard):

\begin{itemize}

\item $\phi = E_a(t_1^\pi,t_2^\pi)$, the path 
$\alpha(\pi)$ is $v_1 a_1 v_2 \cdots v_{n-1} a_{n-1} v_n$,
and it is the case that 
$\alpha(t_2^\pi) = \alpha(t_1^\pi) + 1$ and $a = a_{\alpha(t_1^\pi)}$.  

\item $\phi = t_1^\pi < t_2^\pi$ and $\alpha(t_1^\pi) <
\alpha(t_2^\pi)$.  

\item $\phi = (t_1 \sim t_2)$, $t_1$ is of sort $\pi_1$, $t_2$ is of
  sort $\pi_2$, and $\kappa(v_1) = \kappa(v_2)$, where $v_i$ is the
  node in position $\alpha(t_i)$ of $\alpha(\pi_i)$, for $i = 1,2$.


\item $\phi = \exists t^\pi \psi$ and one of the following holds: 
\begin{enumerate}
\item  $t^\pi$ does not appear free in $\psi$, or
\item  both $t^\pi$ and $\pi$ appear free in $\psi$, and  
there is a position $i$ in
  $\alpha(\pi)$ such that $(G,\alpha[t^\pi \to i]) \models \psi$, or
\item $t_\pi$ appears free in $\psi$, $\pi$ does not appear free in
$\psi$, and there is a path $\rho$ in
  $G$ and  a position $i$ in
  $\rho$ such that $(G,\alpha[\pi \to \rho,t^\pi \to i]) \models \psi$. 
\end{enumerate} 

\item $\phi = \exists \pi \psi$ and the following holds: 
\begin{enumerate} 
\item $\pi$ does not appear free in $\psi$, or 
\item there is a path $\rho$ in
  $G$ such that $(G,\alpha[\pi \to \rho]) \models \psi$. 
\end{enumerate} 

\end{itemize}

\begin{example}
A simple example from \cite{WL} that shows that WL expresses 
NP-complete properties is the following query that checks if a
graph $G$ has a Hamiltonian path:
$$\exists \pi \: \big(\, \forall t_1^\pi \forall t_2^\pi 
\, 
(t_1^\pi \neq t_2^\pi \to 
t_1^\pi \not\sim t_2^\pi)
\, \wedge \, \forall \pi' \forall t_1^{\pi'} \exists t_2^\pi (t_1^{\pi'} \sim t_2^\pi)
\, \big).$$ In fact, this query expresses that there is a path $\pi$
in $G$ that
does not repeat nodes 
(because $\pi$ satisfies $\forall t_1^\pi \forall t_2^\pi 
(t_1^\pi \neq t_2^\pi \to 
t_1^\pi \not\sim t_2^\pi)$), and every node belongs to such
path (because $\pi$ satisfies $\forall \pi' \forall t_1^{\pi'} \exists t_2^\pi
(t_1^{\pi'} \sim t_2^\pi)$, and, therefore, every node that occurs in some path $\pi'$ in
the graph database also occurs in $\pi$). 
Note that this formula uses in an essential way the fact that 
$G$ is a graph database, i.e., that each node is
uniquely identified by its data value. \hfill $\Box$
\end{example} 

\section{WL Evaluation is Non-elementary in Data Complexity}
\label{sec:nonelem}

In this section we pinpoint the precise complexity of query evaluation for WL.
It was proven in \cite{WL} that this problem is decidable. Although the precise 
complexity of this problem was left open in \cite{WL}, one can prove that this 
is, in fact, a non-elementary problem  by an easy
translation from the satisfiability problem for FO formulas -- which is known to
be non-elementary \cite{robertson,stockmeyer}. 
In databases, however, one is often interested in a different measure
of complexity -- called {\em data complexity} \cite{Vardi} -- that
assumes the formula $\phi$ to be fixed. This is a reasonable
assumption since databases are usually much bigger than
formulas. Often in the setting of data complexity the cost of
evaluating queries is much smaller than in the general setting in
which formulas are part of the input. The main result of this section
is that the data complexity of evaluating WL formulas is
nonelementary even \emph{over graph databases}, which rules out its 
practicality.  

Let $\phi$ be a WL formula without free variables. 
The evaluation problem for $\phi$, denoted {\sc Eval}(WL,$\phi$), is 
defined as follows: Given a data graph $G$, is it the case that $G \models \phi$? We prove 
the following: 

\begin{theorem} \label{theo:main} 
The evaluation problem for WL is non-elementary in data
complexity. In particular, for each $k \in \mathbb{Z}_{> 0}$, there is a 
finite alphabet $\Sigma$ and a Boolean formula 
$\phi$ over $\Sigma$,  
such that the problem {\sc Eval}(WL,$\phi$) of evaluating the WL formula $\phi$
is $k$-{\sc Expspace}-hard. In
addition, the latter holds even if the
input is restricted to the class of graph databases.  
\end{theorem} 

We prove the above result by showing that for all natural numbers $k$, the data complexity of the
model checking problem for $\wl$ is $k$-{\sc ExpSpace}-hard. For all
natural numbers $k$ and $f_0$,  
we provide a reduction to the class of problems solvable by a Turing
machine using a tape of size $\tow{k}{f_0 n}$ given an input word of size
$n$, where $\tow{1}{n} := 2^{n}$ and $\tow{k+1}{n}=2^{\tow{k}{n}}$.
 
More precisely, for all natural numbers $k > 0$, there is a Turing
machine $M$ and a constant $f_0$ 
such that the following problem is $k$-{\sc
  ExpSpace}-hard: given a word $w$ of size $n$, is there an accepting
run of $M$ over $w$ using at most $\tow{k}{f_0 n}$ cells? We prove that there is a formula $\phi
\in \wl$ such that for all words $w$ of size $n$, there is a graph
$G_w$ such that
\begin{equation} \label{eq:gw}
G_w \vDash \phi \quad \text{iff} \quad \text{there is an accepting
run of $M$ over $w$ using at most $\tow{k}{f_0 n}$ cells}. 
\end{equation}
Before giving a proof, we sketch the case $k = 1$ here, which illustrates the 
proof idea. 
Let $M$ be a Turing machine $M$ such that the following problem is {\sc
  ExpSpace}-hard: given a word $w$ of size $n$, is there an accepting
run of $M$ over $w$ using at most $2^{f_0 n}$ cells? 
The formula $\phi$ that we will define and satisfying equivalence~\eqref{eq:gw} is 
of the form 
\[
\exists \pi \psi(\pi),
\] 
where $\psi$ is a formula that does not contain any quantification
over path variables. Given a word $w$ of size $n$, the label of the path $\pi$ in the graph
$G_w$ will encode an accepting run of $M$ over the word $w$ in the
following way.

Given a word $w$ of size $n$, consider a configuration $C$ of the run of
$M$ over $w$ where the
head is scanning the cell number $i_0$, the machine is
in state $q$ and the content of the tape is the word $w'=w'_0 \dots
w'_j$ ($j = 2^{f_0 n}-1$).  
We may encode 
the configuration $C$ by the word $e_C= d^C_0 \dots d^C_{j}$ where each $d^C_i$
 encodes the information in cell number $i$ and $j=2^{f_0 n}-1$. More precisely, 
we define $d^C_i$
as a word of the form
\begin{equation} \label{eq:ci}
c(i)  \; (q'_i,w'_i),
\end{equation}
where $c(i)$ and $q'_i$ are defined as follows. The word $c(i)$
is the binary encoding of the number $i$. The
 letter $w'_i$ is the content of the cell $i$. The letter $q'_i$ is equal
 to the dummy symbol $\$$ if the head is not scanning the cell number $i$; 
 otherwise,
 $q'_i$ is equal to the state $q$. That is, $q'_{i_0} = q$
 and for all $i \neq i_0$, $q'_i = \$$. We encode a run $C_0 C_1
 \dots$ as the sequence $e_{C_0} e_{C_1} \dots$.

We think of a path $\pi$ encoding a run as consisting of
two parts: the first part contains the encoding $e_{C_0}$ of the
initial configuration and is a path through a subgraph $I_w$ of $G_w$,
while the second part contains the encoding $e_{C_1}e_{C_2} \dots$ and
is a path through the subgraph $H$ of $G_w$. If $Q$ is the set of
states of $M$ and $\Sigma$ is the alphabet, we define $H$
as the following graph
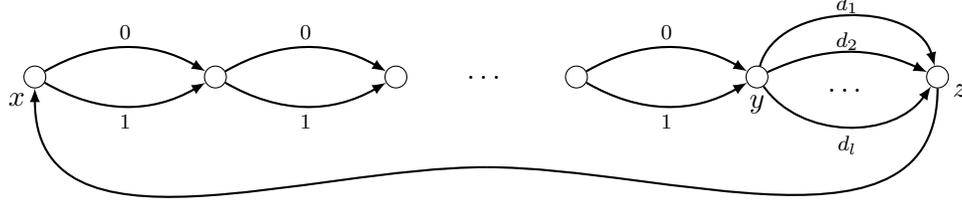
\vspace{-0.3cm}
\begin{center}
\begin{tikzpicture} [scale=0.6]
\node(1) [draw,circle,scale=0.8] at (0,0) {};
\node(2) [draw, circle,scale=0.8] at (4,0) {};
\node at (-0.4,-0.5) {$x$};
\node at (16,-0.6) { $y$};
\node(Y)  at (2,1) {{\scriptsize $0$}};
\node(N) at (2,-1) {{\scriptsize $1$}};
\draw [->,>=latex, thick]
    (1) to[bend left] (2) ;
\draw [->,>=latex, thick]
    (1) to[bend right] (2);

 \node(Y2) at (6,1) {{\scriptsize $0$}};
\node(N2) at (6,-1) {{\scriptsize $1$}};
\node(3) [draw,circle,scale=0.8] at (8,0) {};
\draw [->,>=latex, thick]
    (2) to[bend left] (3); 
\draw [->,>=latex, thick]
    (2) to[bend right] (3);

\node(d) at (10,0) {$\dots$};
\node(f)  [draw, circle,scale=0.8] at (16,0) {};
\node(n)  [draw, circle,scale=0.8] at (12,0) {};
 \node(Yn) at (14,1) {{\scriptsize $0$}};
\node(Nn) at (14,-1) {{\scriptsize $1$}};
\draw [->,>=latex, thick]
    (n) to[bend left] (f); 
\draw [->,>=latex, thick]
    (n) to[bend right] (f); 

\node(z) at (20.5,-0.3) {$z$};

\node(d1) at (18,1.6) {{\scriptsize $d_1$}};
\node(d2) at (18,0.8) {{\scriptsize $d_2$}};
\node(d3) at (18,-1.5) {{\scriptsize $d_l$}};

\node(f2) [draw,circle,scale=0.8] at (20,0) {};
\draw [->,>=latex, thick]
    (f) to[bend left =75] (f2);
\draw [->,>=latex, thick]
    (f) to[bend left=25] (f2);
\draw [->,>=latex, thick]
    (f) to[bend right=55] (f2);
\node(r)[anchor=east] at (10,-2){};
\node(r2)[anchor=west] at (10,-2){};
\node at (18,-0.3) {$\dots$};
\draw [thick]
    (f2) to[out=270,in=0] (r);
\draw [->,>=latex, thick] (r2) to[out=180,in=270] (1);
\end{tikzpicture}
\end{center}
\vspace{-0.5cm}
where $l$ is equal to $|(Q \cup \{\$\}) \times \Sigma|$,  $\{ d_i : 1\leq i \leq l\} = (Q \cup \{\$\}) \times \Sigma$ and the number of nodes with outgoing edges with labels $0$ and $1$
is equal to $f_0 n$. The label of a path $\pi'$ from the ``left-most'' node $x$ to
the ``right-most'' node $z$ with only once occurrence of $x$ is
exactly the description of a cell in a configuration: it is the binary
encoding of a natural number $< 2^{f_0 n}$ followed by a pair of the form 
$(q',a)$. We can define a
formula $\phi_C \in $ WL such that for all paths $\pi$ starting in
$x$ and ending in $z$,
\[
H \vDash \phi_C (\pi) \quad \text{iff} \quad \text{the label of
  $\pi$ is the encoding of a configuration.}
\]
We do not give details; $\phi_C$ has to express that the
encoding
of a configuration only has one tape head, that the first number
encoded in binary is $0$, that the last number is $2^{f_0 n}-1$ and that the
encoding of the description of cell number $j$ is followed by the
description of cell number $j+1$. Using the formula $\phi_C$, we can define
a formula $\phi_1$ such that for all paths $\pi$,
\[
H \vDash \phi_1 (\pi) \quad \text{iff} \quad \text{the label of
  $\pi$ is the encoding of an accepting run.}
\]
The formula $\phi_1$ has to ensure that if $e_C e_{C'}$ occurs in the
label of $\pi$, then $C$ and $C'$ are consecutive configurations
according to $M$. Moreover, $\phi_1$ has to express that eventually we
reach the final state. In order to express $\phi_C$ and $\phi_1$, we use the ability of WL to check whether two positions
correspond to the same node. For example, in order to define $\phi_1$,
since we need to compare consecutive configurations $e_C$ and
$e_{C'}$, we need to be able to compare the content of a cell in
configuration $C$ and the content of that same cell in $C'$. In
particular, we want to be able to express whether two subpaths $\pi'_0$
and $\pi'_1$ of $\pi$
starting in $x$ and ending in $y$ correspond to the binary encoding of
the same number. Since the length of such subpaths depends on $n$,
we cannot check node by node whether the two subpaths are
equal. However, it is sufficient to check that if $t_0^{\pi'_0}$ and
$t_1^{\pi'_1}$ correspond to the same node ($t_0^{\pi'_0} \sim t_1^{\pi'_1}$), then their successors
also correpond to the same node ($t_0^{\pi'_0}+1 \sim
t_1^{\pi'_1}+1$).  Note
that using the facts that $\pi'_0$ and $\pi'_1$ are subpaths of $\pi$,
we will be able to define $\phi_1$ such that it only contains
quantifications over node variables (and no quantifications over path variables). Similarly, in the formula $\phi_C$, we use the
operator $\sim$ in order to express that two subpaths correspond to the
binary encodings of numbers that are successors of each other. 

Similarly to the way we define the graph $H$, we can introduce a graph $I_w$ and a formula
$\phi_0(\pi)$ such that 
\[
I_w \vDash \phi_0 (\pi) \quad \text{iff} \quad \text{the label of 
  $\pi$ is the encoding $e_{C_0}$,}
\]
where $C_0$ is the initial configuration of the run of $M$
 over $w$. By adding an edge from $I_w$ to $H$, we 
 construct a graph $G_w$ such that for all paths $\pi$, $
G_w \vDash \phi_0(\pi) \wedge \phi_1(\pi)$  iff the label of 
  $\pi$ is the encoding of an accepting run over $w$. 
Hence, the formula $\phi:= \exists \pi (\phi_0(\pi) \wedge \phi_1(\pi)
)$ satisfies~\eqref{eq:gw}.

For the case where $k>1$, the problem to adapt the above proof is that
we have to consider Turing machine configurations whose size is bounded by a 
tower of
exponentials of height $k$. If $k>1$, the binary  representation of
such a bound is not polynomial. The trick is to 
represent such exponential towers 
by $k$-counters. A $1$-counter is the binary representation of a
number. If $k>1$, a $k$-counter  is a word $\sigma_0 l_0 \dots
\sigma_{j_0} l_{j_0}$, where $l_{j}$ is a $(k-1)$-counter and $\sigma_j
\in \{ 0,1\}$.

\begin{definition*}
For all natural numbers $k$, we consider the alphabet $\Sigma_k = \{
a_k,b_k\}$, where $a_k$ and $b_k$ represent $0$ and $1$ respectively. We
define $\Gamma_k$ as the alphabet $\Sigma_1 \cup \dots \cup \Sigma_k$.

A {\em $1$-counter} of length $n$ is a sequence of the form
\[
l_0 \dots l_{f_0 n-1},
\]
where for all $ 0 \leq i < f_0 n$, $l_i \in \Sigma_1$. This $1$-counter
represents the number $\sum_{i=0}^{f_0 n-1} l_i 2^i$. Recall that if $l_i$
is equal to $a_1$ (resp. $b_1$), then $l_i$ represents $0$ (resp. $1$).

If $k\geq 2$, a {\em $k$-counter} of length $n$ is a sequence of the form
\[
\sigma_0 l_0 \dots \sigma_j l_j,
\]
where for all $0 \leq i \leq j$, $l_i \in \Sigma_k$, $\sigma_i$ is a
$(k-1)$-counter representing the number $i$ and $j = \tow{k-1}{f_0 n}-1$. This $k$-counter
represents the number $\sum_{i=0}^{j} l_i 2^i$. Again recall that  if $l_i$
is equal to $a_1$ (resp. $b_1$), then $l_i$ represents $0$ (resp. $1$).

\OMIT{
A {\em $(k,f_0 n)$-description} over an alphabet $\Delta$ is a sequence 
\[
\sigma_0 d_0\dots \sigma_j d_j,
\]
where for all $0 \leq i \leq j$, $d_i \in \Delta$, $\sigma_i$ is a
$(k-1)$-counter representing the number $i$ and $j =
\tow{k}{c(n-1)}-1$.
}

A {\em $(k, f_0 n,p)$-description} (over an alphabet $\Delta$) is a sequence 
\[
\sigma_{p} d_{p} \dots \sigma_j d_j ,
\]
where for all $p \leq i \leq j$, $d_i \in \Delta$, $\sigma_i$ is a
$(k-1)$-counter representing the number $i$ and $j =
\tow{k}{f_0 (n-1)}-1$. 
A {\em $(f_0 k,n)$-description} (over an alphabet $\Delta$) is a 
$(k, f_0 n,0)$-description.
\qed
\end{definition*}

Note that a $(k,f_0 n)$-description over the alphabet $\Sigma_k$ is a
$k$-counter of length $n$. 
If $\Delta$ is the alphabet $ (Q \cup \{ \$
\}) \times \Sigma$ (where $Q$ is the set of states and $\Sigma$ is the
alphabet of the machine), a $(k,f_0 n)$-description over $\Delta$ is of the form
\[
l_0 (x_0,y_0)  \dots l_j (x_j,y_j) 
\]
where $j =\tow{k}{f_0 n}-1$. Hence, if we define $c(i)$ in~\eqref{eq:ci} as
the $k$-counter encoding the number $i$, the encoding of a
configuration (as defined above) is nothing but a
$(k,f_0 n)$-description. 

In particular, if we want to encode a run as the label of a path
satisfying some well-chosen formula in a well-chosen graph, we should
also 
be able to encode $(k,f_0 n,p)$-descriptions as labels of paths. We show
how to do so in the
following lemma.

\begin{notation*}
Given a path $\pi$ in a graph over an alphabet $\Delta$, we denote by
$l(\pi)$ the label of $\pi$. Given an alphabet $\Delta' \subseteq
\Delta$, we denote by $l_{\Delta'} (\pi)$ the trace of $l(\pi)$ over
the alphabet $\Delta'$, that is, the subsequence of $l(\pi)$ obtained
by deleting the letters that do not belong to $\Delta'$. 

Let $G'=(V',E',\kappa')$ be a subgraph of
$G=(V,E,\kappa)$ and let $\pi$ be a path in $G$ and of the form
\[
v_1 a_1 v_{2} \dots v_{n-1} a_n v_n,
\]
where $(v_{i},a_{i},v_{i+1}) \in E$ for all $1 \leq i <n$. 
Assume that there are $i_0$ and $i_1$ such that $i_0 \leq i_1$ and 
\[
\{ v_i : v_i \in V', 1 \leq i \leq n\} = \{ v_{i_0}, \dots, v_{i_1} \},
\]
that is, once the path leaves $G'$, it never goes back to $G'$. 
Then we define the {\em trace} $\pi'$ of $\pi$ on $G'$ as the subpath
\[
v_{i_0} a_{i_0} v_{i_0+1} \dots v_{i_1-1} a_{i_1} v_{i_1},
\]
that is, $\pi'$ is the longest subpath of $\phi$ with nodes in $G'$.

In order to make notation easier, we also abbreviate the formula
\[
\exists s \; E_a(s,t)
\] 
by $a(t)$.

Given a formula $\phi(\pi,s,t)$ with path variable $\pi$ and node variables $s$ and $t$, we denote by $\phi(\pi_{s,t})$ the formula obtained by replacing in $\phi(\pi,s,t)$ each quantification of the form
\[
\exists r^\pi
\] 
by
\[
\exists r^\pi \text{ s.t. } (s < r <t).
\]
Intuitively, we ``restrict'' the path $\pi$ to the nodes occurring between $s$ and $t$.
\end{notation*}

\begin{lemma} \label{lem:pass}
For all $n$ and $k$ and for all alphabets $\Delta$, there are 
formulas $\phi_{k,n,p}^{\Delta}(\pi)$ ($0 \leq p <n$) and a graph 
$G^{\Delta}_{k, n}$ satisfying the following. There is a unique node with
an outgoing (resp. incoming) edge with label $i_{k,n}^\Delta$ (resp. $f_{k,n}^\Delta$); moreover, that node has
no incoming (resp. outgoing) edge.  That node is called the
{\em initial} (resp. {\em final}) node. 
Finally,   $G^{\Delta}_{k,n} \vDash \phi_{k,n,p}^{\Delta}(\pi)$ iff
 the label $l(\pi)$ of $\pi$ satisfies the following conditions:
\begin{itemize}
\item only the first edge of $\pi$ is labeled $i_{k,n}^\Delta$,
\item only the last edge of $\pi$ is labeled $f_{k,n}^\Delta$,
\item if $k\geq 2$ and $\Delta'=\Delta \cup \Gamma_{k-1}$, then
$l_{\Delta'}(\pi)$  is a $(k,f_0 n,p)$-description over $\Delta$;  
\item if $k =1$, $l_{\Sigma_1}(\pi)$ is a $1$-counter of length $n$.
\end{itemize}
We let $\phi_{k,n}^{\Delta}(\pi)$  be an abbreviation for $\phi_{k,n,0}^{\Delta}(\pi)$. 

Moreover, if $\Delta=\Sigma_k$, then there are formulas
$\mathit{succ}_{k,n}(\pi,\pi')$, $\mathit{number}^i_{k,n}$ ($1 \leq i \leq
n$), $\mathit{last}_{k,n}$ and $\mathit{eq}_{k,n}(\pi,\pi')$  such that for all paths $\pi$ and
$\pi'$ satisfying $G^{\Delta}_{k,n} \vDash \phi_{k,n}^{\Delta}(\pi) \wedge \phi_{k,n}^{\Delta}(\pi')$,  we have
\begin{itemize}
\item $G^{\Delta}_{k,n} \vDash \mathit{succ}_{k,n}(\pi,\pi')$ iff the number encoded by $l_{\Gamma_k}(\pi')$ is the
  successor of the number encoded by $l_{\Gamma_k} (\pi)$.
\item $G^{\Delta}_{k,n} \vDash \mathit{number}^i_{k,n}(\pi)$ iff $l_{\Gamma_k} (\pi)$ is the
  encoding of the number $i$. 
\item $G^{\Delta}_{k,n} \vDash \mathit{last}_{k,n}(\pi)$ iff $l_{\Gamma_k} (\pi)$ is the
  encoding of the number $\tow{k}{f_0 n}$. 
\item $G^{\Delta}_{k,n} \vDash \mathit{eq}_{k,n}(\pi,\pi')$ iff the
  number encoded by $l_{\Gamma_k} (\pi')$ is equal to the
  number encoded by $l_{\Gamma_k} (\pi)$.
\end{itemize}
\end{lemma}

\begin{proof}
The formulas and the graph are defined by induction on $k$. Suppose
first that $k=1$ and $\Delta=\Sigma_1$. We define $G_0$ as the following graph
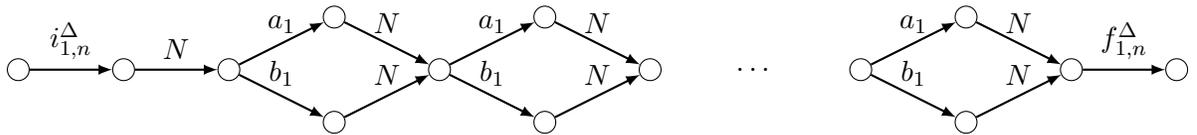
\begin{center}
\begin{tikzpicture} [scale=0.7]
\node(ss) [draw, circle,scale=0.8] at (0,0) {};
\node(S) [draw, circle,scale=0.8] at (2,0) {};
\draw [->,>=latex, thick]
    (ss) -- (S) node[midway,above]{$i^\Delta_{1,n}$}; 

\node(1) [draw,circle,scale=0.8] at (4,0) {};
\node(Y) [draw, circle,scale=0.8] at (6,1) {};
\node(N) [draw,circle,scale=0.8] at (6,-1) {};
\draw [->,>=latex, thick]
    (1) -- (Y) node[midway,above]{$a_1$}; 
\draw [->,>=latex, thick]
    (1) -- (N) node[midway,above]{$b_1$}; ; 
\draw [->,>=latex, thick]
    (S) -- (1) node[midway,above]{$N$}; 

\node(2) [draw, circle,scale=0.8] at (8,0) {};
 \node(Y2) [draw, circle,scale=0.8] at (10,1) {};
\node(N2) [draw,circle,scale=0.8] at (10,-1) {};
\draw [->,>=latex, thick]
    (2) -- (Y2) node[midway,above]{$a_1$};  
\draw [->,>=latex, thick]
    (2) -- (N2) node[midway,above]{$b_1$}; 
\draw [->,>=latex, thick]
    (Y) -- (2) node[midway,above]{$N$}; 
\draw [->,>=latex, thick]
    (N) -- (2) node[midway,above]{$N$}; 

\node(3) [draw,circle,scale=0.8] at (12,0) {};
\draw [->,>=latex, thick]
    (Y2) -- (3) node[midway,above]{$N$}; 
\draw [->,>=latex, thick]
    (N2) -- (3) node[midway,above]{$N$}; 

\node(d) at (14,0) {$\dots$};

\node(n)  [draw, circle,scale=0.8] at (16,0) {};
 \node(Yn) [draw, circle,scale=0.8] at (18,1) {};
\node(Nn) [draw,circle,scale=0.8] at (18,-1) {};
\draw [->,>=latex, thick]
    (n) -- (Yn) node[midway,above]{$a_1$};
\draw [->,>=latex, thick]
    (n) -- (Nn) node[midway,above]{$b_1$};
\node(f)  [draw, circle,scale=0.8] at (20,0) {};
\node(ff) [draw, circle,scale=0.8] at (22,0) {};
\draw [->,>=latex, thick]
    (f) -- (ff) node[midway,above]{$f^\Delta_{1,n}$};
\draw [->,>=latex, thick]
    (Yn) -- (f) node[midway,above]{$N$}; 
\draw [->,>=latex, thick]
    (Nn) -- (f) node[midway,above]{$N$}; 
\end{tikzpicture}
\end{center}
where the number of nodes with outgoing edges with labels $a_1$ and
$b_1$, is equal to $f_0 n$. The label $N$ is an additional label that we introduce in order to simplify the notation in the
formulas.

We let $G^\Delta_{1,n}$ be the graph $G_0$. We define now the formula
$\phi_{1,n}^{\Delta}$. In fact, any path $\pi$ over $G_0$ starting
with the node with no incoming edge and ending  with the node with no
outgoing edge, will be such that $l_{\Sigma_1}(\pi)$ is the encoding of a $1$-counter. Hence, we can define $\phi_{1,n}^{\Delta}$ as the conjunction of the formula
\[
\exists s^\pi [\neg \exists t^\pi,  t< s] 
\]
and the formula
\[
\exists s^\pi  [\neg \exists t^\pi,  s <t ] .
\]
We show now how to define the formulas
$\mathit{num}^i_{1,n}(\pi)$ (by induction on $i$),
$\mathit{eq}_{k,n} (\pi,\pi')$ and $\mathit{last}_{1,n}(\pi)$. For the
formula  $\mathit{last}_{1,n}(\pi)$, a path $\pi$ corresponds to the
encoding of the number $2^{f_0 n}-1$ iff we always choose the node with label $b_1$. Or equivalently, if we never choose the node with label $a_1$. Hence, we may define $\mathit{last}_{1,n}(\pi)$ as the formula
\[
\neg \exists s^\pi, a_1(s). 
\]
For the formula $\mathit{eq}_{1,n}(\pi,\pi')$, two paths $\pi$ and
$\pi'$ correspond to the same number iff $\pi$ and $\pi'$ are
equal. Since $\pi$ and $\pi'$ are simple paths with the same starting
node, this is equivalent over graph databases (where each node carries
a different data value) to the fact the the following formula holds
\[
\forall t^\pi, (t')^{\pi'} [t \sim t' \to (t+1) \sim (t'+1)].
\]
The formulas $\mathit{num}^i_{1,n}(\pi)$ is defined by induction
on $i$. If $i=0$, the path $\pi$ encodes the number $0$ iff we always
choose the node with label $a_1$. Or equivalently, if we never choose
the node with label $b_1$, which is expressed by
\[
\neg \exists s^\pi, b_1(s). 
\]
For the induction case, the path $\pi$ encodes the number $i+1$ iff there is a path $\pi''$ encoding the number $i$ and the number encoded by $\pi$ is the successor of the  number encoded by $\pi''$. Hence, we can define $\mathit{num}^{i+1}_{1,n}(\pi)$ as the formula
\[
\exists \pi'' (\mathit{num}^i_{1,n}(\pi'') \wedge \mathit{succ}_{1,n} (\pi'',\pi)).
\]
In order to finish the base case, it remains to define the formula $\mathit{succ}_{1,n} (\pi,\pi')$.  Basically, we have to simulate addition in binary. If
$x_1 \dots x_{f_0 n}$ is the binary encoding of a number $i < 2^{f_0 n}-1$, then the
binary encoding of the number $i+1$ is the sequence $x'_1 \dots
x'_{f_0 n}$
such that $x_m$ is equal to 
\begin{enumerate}[label=\({\alph*}]
\item $1$ if $x_m =0$ and all the elements $x_{m+1}, \dots,x_{f_0
    n}$ are
  equal to $1$,
\item $0$ if $x_m =1$ and all the elements $x_{m+1}, \dots,x_{f_0
    n}$ are
  equal to $1$,
\item $0$ if $x_m =0$ and there is an element in the sequence
  $x_{m+1} \dots x_{f_0 n}$ that is equal to $0$, 
\item $1$ if $x_m =1$ and there is an element  in the sequence
  $x_{m+1} \dots x_{f_0 n}$ that is equal to $0$.
\end{enumerate}
Case~(a) can be expressed by the following formula
\[
\forall t^\pi,(t')^{\pi'} [t \sim t' \wedge a_1 (t) \wedge  \forall s \in \pi [(t < s) \wedge N(s-1) \to b_1(s)]] \to b_1(t').
\]
The other cases can be treated similarly. This finishes the base case.

We turn now to the induction step. If $\Delta=\{ d_1,\dots,d_l\}$, we define $G_{k+1,n}^{\Delta}$ as
the following graph
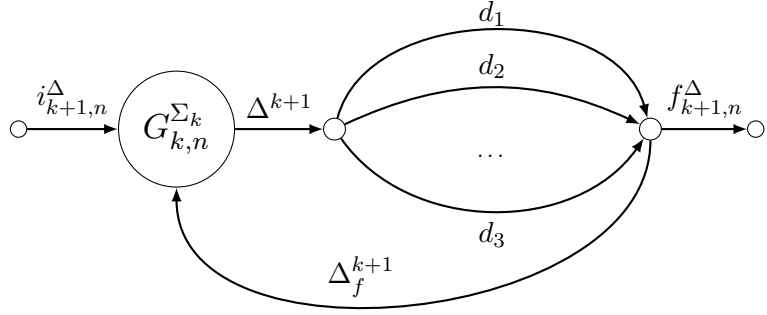
\begin{center}
\begin{tikzpicture} [scale=0.7]
\node(S) [draw, circle,scale=0.6] at (0,0) {};
\node(G) [draw, circle,scale=1.2] at (3,0) {$G^{\Sigma_k}_{k,n}$};
\node(Si) [draw, circle,scale=0.8] at (6,0) {};
\node(a1)  at (9,2.2) {$d_1$};
\node(a2)  at (9,1.2) {$d_2$};
\node(al) at (9,-2) {$d_3$};
\node [scale=0.8] at (9,-0.5) {$\dots$};
\node(b) [draw, circle,scale=0.8] at (12,0) {};
\node(F) [draw, circle,scale=0.6] at (14,0) {};
\node at (6.5,-2.8) {$\Delta_f^{k+1}$};

\draw [->,>=latex, thick]
    (S) -- (G)
node[midway,above]{$i_{k+1,n}^\Delta$}; 
\draw [->,>=latex, thick]
    (G) -- (Si)
node[midway,above]{$\Delta^{k+1}$}; 
\draw [->,>=latex, thick]
    (Si) to[bend left =75] (b);
\draw [->,>=latex, thick]
    (Si) to[bend left=25] (b);
\draw [->,>=latex, thick]
    (Si) to[bend right=55] (b);
\draw [->,>=latex, thick]
    (b) -- (F) node[midway,above] {$f_{k+1,n}^\Delta$};   
\draw [->,>=latex, thick]
    (b) edge[out=270,in=270] (G);

\end{tikzpicture}
\end{center}
The edge with label $i_{k+1,n}^\Delta$ and the edge with label
$\Delta^{k+1}_f$ are pointing to the initial node in $G_{k,n}^{\Sigma_k}$. The edge with label $\Delta^{k+1}$ is an edge
starting from the final node in $G_{k,n}^{\Sigma_k}$.

We define now the formula $\phi^\Delta_{k+1,n,p}(\pi)$. The intuition
is as follows. We encode a $(k+1,f_0n,p)$-description 
\[
\sigma_{p} d_{p} \dots \sigma_j d_j, 
\]
as a path $\pi$ starting with the edge with label $i_{k+1,n}^\Delta$
and ending with the edge with label $f_{k+1,n}^\Delta$. Each
$k$-counter $\sigma_i$ will correspond to a path through the subgraph
$G_{k,n}^{\Sigma_k}$, while $d_i$ will correspond to the label of an
edge occurring after the edge with label $\Delta^{k=1}$. The formula  $\phi^\Delta_{k+1,n,p}(\pi)$ needs to ensure that the following hold:
\begin{enumerate}[label=\({\alph*}]
\item The first  edge of $\pi$ is the edge with label $i_{k+1,n}^\Delta$. 
\item Each ``passage'' of the path $\pi$ through the graph $G_{k,n}^{\Sigma_k}$ corresponds to the encoding of a $k$-counter. To express this, we will use the formula $\phi_{k,n}^{\Sigma_k} (\pi)$ given by the induction hypothesis.
\item  The first time the path $\pi$ ``goes through'' the graph $G_{k,n}^{\Sigma_k}$  corresponds to the encoding of the number $p$.
\item Two successive ``passages'' of $\pi$ through the graph
  $G_{k,n}^{\Sigma_k}$   correspond to two successive $k$-counters.
\item The edge with label $f_{k+1,n}^\Delta$ occurs after the edge 
  with label $\Delta_f^{k+1}$  iff the last passage of the path $\pi$
  through the graph    $G_{k,n}^{\Sigma_k}$, corresponded to the
  encoding of the number $\tow{k}{f_0 n}$. This ensures that we fully
  encode a $(k+1,f_0 n)$-description, and not a subsequence of it.
\end{enumerate}

\noindent We only show how to express~(b) as this is one of the most difficult
cases and the other ones can be treated similarly. 

For ~(b) we have to express that each passage of $\pi$ through the
  graph $G_{k,n}^{\Sigma_k}$ corresponds to the encoding of a
  $k$-counter. Recall that by the induction hypothesis, since a
  $(k,f_0 n)$-description over $\Sigma_k$ is a $k$-counter of length $k$,
  the formula $\phi_{k,n}^{\Sigma_k} (\pi')$ is true in the graph
  $G_{k,n}^{\Sigma_k}$ iff $l_{\Gamma_k}(\pi')$ is the encoding of a
  $k$-counter of length $n$.

Hence, in order to express~(b), it is enough to ensure that if $s$ is
the first node of a passage of $\pi$ through  $G_{k,n}^{\Sigma_k}$
and if $t$ is the last node of that same passage, then the formula
$\phi_{k,n}^{\Sigma_k} (\pi_{s,t})$ holds. We introduce a formula  
  $\mathit{IF}_{k,n}(s,t,\pi)$ such that 
\begin{align} 
\mathit{IF}_{k,n}(s,t,\pi) \text{ holds } & \quad \text{iff} \quad & \text{$s$
is the first node of a passage of $\pi$ through
  $G_{k,n}^{\Sigma_k}$} \label{eq:pass2}  \\ & & 
\text{and $t$ is the last node of that same passage.} \notag
\end{align}
We define $\mathit{IF}_{k,n}(s,t,\pi)$ as the formula 
\[
i ^{\Sigma_k}_{k,n} (s) \wedge E_{f ^{\Sigma_k}_{k,n}}(t-1,t) \wedge (s<t) \wedge \neg \exists u^\pi
\; [(s <u < t) \wedge \bigvee_i d_i(u) ].
\]
This equivalent to saying that $s$ is the initial node of
$G_{k,n}^{\Sigma_k}$, that $t$ is the final node of
$G_{k,n}^{\Sigma_k}$ and the path ``never goes out'' of the graph
$G_{k,n}^{\Sigma_k}$ (this can be enforced by imposing that we do not
go through the edge with label $d_i$ for some $i$). 

We define now the formula $\chi_1(s,t,\pi)$ expressing condition~(b),
that is, if  $s$ is
the first node of a passage of $\pi$ through  $G_{k,n}^{\Sigma_k}$
and if $t$ is the last node of that same passage, then the formula
$\phi_{k,n}^{\Sigma_k} (\pi_{s,t})$ holds. By~\eqref{eq:pass2}, we may
define 
$\chi_1(s,t,\pi)$  as the formula
\[
\mathit{IF}_{k,n} (s,t,\pi) \to \phi_{k,n}^{\Sigma_k} (\pi_{s,t}).
\]\smallskip

\noindent We turn now to the definitions of the formulas $\mathit{succ}_{k+1,n}(\pi)$,
$\mathit{num}^i_{k+1,n}(\pi)$, $\mathit{eq}_{k+1,n} (\pi,\pi')$ and
$\mathit{last}_{k+1,n}(\pi)$. The formulas $\mathit{succ}_{k+1,n}(\pi)$,
$\mathit{num}^i_{k+1,n}(\pi)$ and
$\mathit{last}_{k+1,n}(\pi)$ are defined in a similar fashion as the
basis case ($k=1$). 

In order to define the formula $\mathit{eq}_{k+1,n}
(\pi,\pi')$, let $\pi$ and $\pi'$ be two paths satisfying the formula
$\phi_{k+1,n}^{\Sigma_{k+1}}$. Recall that $\pi$ corresponds to the
encoding of a $(k+1)$-counter 
\[
\sigma_1 d_1 \dots \sigma_j d_j,
\]
where each $\sigma_i$ corresponds to a passage $\pi_{s,t}$ of $\pi$ through
$G_{k,n}^{\Sigma_k}$ and $d_i$ corresponds to the label of an edge
occurring right after that passage. Given the structure of the graph
$G_{k+1,n}^{\Delta}$, that edge is the incoming edge of the node $t+2$.

The paths $\pi$ and $\pi'$ correspond to the encoding of
the same $(k+1)$-counter if for all passages $\pi_{s,t}$ of $\pi$ through
$G_{k,n}^{\Sigma_k}$ and for all  passages $\pi_{s',t'}$ of $\pi$ through
$G_{k,n}^{\Sigma_k}$ such that $\pi_{s,t}$  and $\pi_{s',t'}$ encode
the same $k$-counter, we have that $t+2$ and $t'+2$ are the same nodes. By~\eqref{eq:pass2} and by the induction hypothesis, this can be expressed by the following
formula   $\mathit{eq}_{k+1,n}
(\pi,\pi')$ given by
\[
\forall s^\pi ,t^\pi \; \forall (s')^{\pi'},(t')^{\pi'} \; [\mathit{IF} (s,t,\pi) \wedge \mathit{IF}(s',t',\pi') \wedge
\mathit{eq}_{k,n} (\pi_{s,t},\pi_{s',t'}) \to( t+2 \sim t'+2)]
\]
This finishes the proof of Lemma~\ref{lem:pass}.
\end{proof}

We are now ready to prove Theorem~\ref{theo:main}.

\medskip 

\noindent {\em Proof of Theorem~\ref{theo:main}.}  
As explained earlier, we prove that for all
Turing machines $M$ and for all $k$, there is a formula $\phi
\in \wl$ such that for all words $w$ of size $n$, there is a graph
$G_w$ such that
\[
G_w \vDash \phi \quad \text{iff} \quad \text{there is an accepting
run of $M$ over $w$ using at most $\tow{k}{n}$ cells}. 
\] 
Let $(\Sigma,Q, \delta,q_i,q_f)$ be the Turing machine $M$,
where $\Sigma$ is the input alphabet together with a blank symbol $B$, $q_0$ is the initial state, $q_f$ is the final state and
$\delta: Q \times \Sigma \to (Q \times \Sigma \times \{ L,R\})$ is the
transition map, where $L$ stands for ``left'' and $R$ stands for ``right''.

The formula $\phi$ is a formula
of the form
\[
\exists \pi \psi(\pi),
\]
where $\psi$ is a formula that does not contain any quantification
over path variables. Given a word $w$, the label of  $\pi$ in the graph
$G_w$ is the encoding of an accepting run of $M$ over the word
$w$. Recall that we encode a configuration of the machine in the
following way. Suppose that $C$ is a configuration 
where the content of the tape is the word $w'=w'_0 \dots w'_j$ ($j =
\tow{k}{f_0 n}-1$), the
head is scanning the cell number $i_0$ and the machine is
in state $q$.  We may encode 
$C$ by the word $e_C= d^C_0 \dots d^C_{j}$ where each $d_i$ is a sequence
\[
c(i)  \; (q'_i,w'_i),
\]
and $c(i), w'_i$ and $q'_i$ are defined as follows. The word $c(i)$
is the $k$-counter encoding the number $i$. The
 letter $w'_i$ is the content of the cell $i$. The letter $q'_i$ is equal
 to $\$$ if the head is not scanning the cell number $i$; otherwise,
 $q'_i$ is equal to the state $q$. 
This implies that given a configuration $C$, the word $ec_C$ is a
$(k+1,n)$-counter over the alphabet $\Delta := (Q \cup \{ \$\}) \times
\Sigma$.
  
The run of $M$ over the word $w$ is a sequence of configurations of
the form $C_0 C_1 \dots$. We encode the run as  the word $e_{C_0} 
e_{C_1} \dots$ (which is a sequence of $(k+1,n)$-counters). We will define the formula $\psi(\pi)$ and the
graph $G_w$ in such a
way that a path $\pi$ satisfies $\psi$ iff the projection of the label
of $\pi$ on the alphabet $\Gamma_k \cup \Delta$ is the
encoding of an accepting run of $M$ over $w$.    

We think of a path $\pi$ encoding a run of $M$ over $w$ as consisting
of two parts. The label of the first part contains the encoding $ e_{C_0}$ of the
initial configuration $C_0$. The label of the second part contains the encoding  $e_{C_1} 
e_{C_2} \dots$ of the remaining part of the run. The first part of the
path $\pi$ is a path in a subgraph $I_w$ of  $G_w$, while the second
part is a path in the subgraph $H$ (independent of $w$) of $G_w$. The graph $G_w$ will be
obtained by adding an edge from a node of $I_w$ to a node of $H$. 

We start by defining the graph $H$. Recall that $\Delta$ 
is the alphabet $ (Q \cup \{ \$ \}) \times \Sigma$. The graph $H$ is defined
as the graph $G^{\Delta}_{k,n}$ with an additional edge from
the final node to the initial node. Hence,
it follows from the proof of Lemma~\ref{lem:pass} that $H$ is the
following graph,

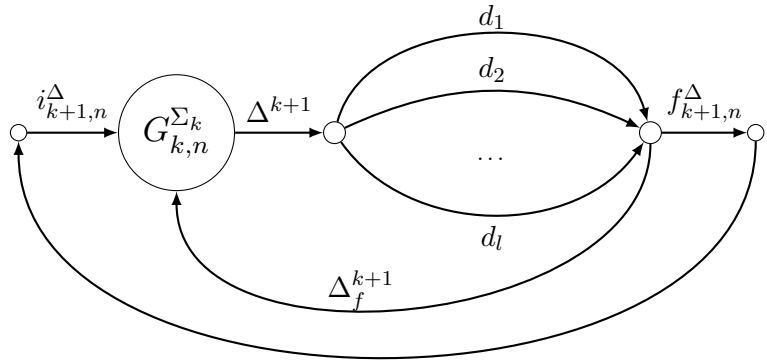
\begin{center}
\begin{tikzpicture} [scale=0.7]
\node(S) [draw, circle,scale=0.6] at (0,0) {};
\node(G) [draw, circle,scale=1.2] at (3,0) {$G^{\Sigma_k}_{k,n}$};
\node(Si) [draw, circle,scale=0.8] at (6,0) {};
\node(a1)  at (9,2.2) {$d_1$};
\node(a2)  at (9,1.2) {$d_2$};
\node(al) at (9,-2) {$d_l$};
\node [scale=0.8] at (9,-0.5) {$\dots$};
\node(b) [draw, circle,scale=0.8] at (12,0) {};
\node(F) [draw, circle,scale=0.6] at (14,0){};
\node at (6.5,-3) {$\Delta^{k+1}_f$};

\draw [->,>=latex, thick]
    (S) -- (G) node[midway, above] {$i_{k+1,n}^\Delta$};  
\draw [->,>=latex, thick]
    (G) -- (Si)
node[midway,above]{$\Delta^{k+1}$}; 
\draw [->,>=latex, thick]
    (Si) to[bend left =75] (b);
\draw [->,>=latex, thick]
    (Si) to[bend left=25] (b);
\draw [->,>=latex, thick]
    (Si) to[bend right=55] (b);
\draw [->,>=latex, thick]
    (b) -- (F) node[midway,above]{$f_{k+1,n}^\Delta$} ;   
\draw [->,>=latex, thick]
    (b) edge[out=270,in=270] (G);

\draw [->,>=latex, thick]
    (F) edge[out=270,in=270] (S);
\end{tikzpicture}
\end{center}
where $\Delta=\{d_1,\dots,d_l\}$ and  where the edges with label
$i_{k,n}^{\Delta}$ and $\Delta_f^{k+1}$ are edges pointing to the
initial node of $G_{k,n}^{\Sigma_k}$ and the edge with label
$\Delta^{k+1}$ is an edge starting from the final node of 
$G_{k,n}^{\Sigma_k}$. In the above paragraphs, any edge pointing to
the graph $G_{k-1,n}^{\Sigma_k}$  is an edge pointing to the initial
node of that graph. Similarly, any edge starting from the graph $G_{k-1,n}^{\Sigma_k}$ will
always be referring to an edge starting in the final node of the graph.

Recall that if a path
$\pi$ encodes a run $C_0C_1C_2 \dots$, the trace of $\pi$ on $H$
will encode the part $C_1 C_2 \dots$ of the run. Each configuration $C_i$ 
is encoded as a $(k+1,f_0 n)$-description over $\Delta$, which will correspond, as
in Lemma~\ref{lem:pass}, to a passage of the path $\pi$ from the
initial node to the final node of $G_{k-1,n}^{\Sigma_k}$.

We define now the graph $I_w$ encoding the initial configuration of
the tape. Recall that in the initial configuration, the tape contains
the word $w=w_0 \dots w_{n-1}$, all the cells with number $\geq n$
contain 
the blank symbol $B$, the head is scanning the first cell and the
state is $q_0$. The graph $I_w$ is obtained by ``assembling'' the
subgraphs $K_0, \dots,K_{n-1}$ and $K$, which we will define next. For each $i \leq n$, the 
graph $K_i$ is such that the label of its unique maximal path  is the
encoding of the cell number $i$ in the initial configuration. The trace
of the path $\pi$ on the graph $K$ is the encoding of the contents of
the cells with number $\geq n$ in the initial configuration.  

More precisely, we define the graphs $K_0, \dots,K_{n-1}$ and $K$  in the
following way. The following graph is the graph $K_0$

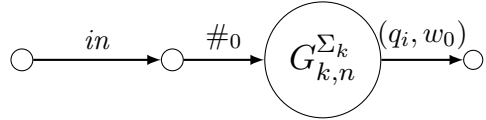
\begin{center}
\begin{tikzpicture} [scale=0.5]
\node(E0) [draw, circle,scale=0.8] at (-8,0) {};
\node(E) [draw, circle,scale=0.8] at (-4,0) {};
\node(T) [draw, circle,scale=1.2] at (0,0) {$G_{k,n}^{\Sigma_k}$};
\node(S) [draw, circle,scale=0.7] at (4,0) {};
\draw [->,>=latex, thick]
    (E0) -- (E) node[midway,above]{$\mathit{in}$};
\draw [->,>=latex, thick]
    (E) -- (T) node[midway,above]{$\#_0$};
\draw [->,>=latex, thick]
    (T) -- (S) node[midway,above]{$(q_i,w_0)$};
\end{tikzpicture}
\end{center}
The node with label $\mathit{in}$ will be the starting node of the path
$\pi$. Since the trace of $\pi$ in $K_1$ encodes the content
of the cell with number $0$ in the initial configuration, its label must contain
the $k$-counter encoding the number $0$ followed by the letter
$(q_i,w_0)$ of the alphabet $\Delta$ (indicating that the first cell
contains the letter $w_0$, the head is scanning the first cell and the
current state is $q_0$).  Using Lemma~\ref{lem:pass} and the formula
$\mathit{num}_{k,n}^0$, we will impose that the passage of $\pi$
through the subgraph $G_{k,n}^{\Sigma_k}$ of 
$K_0$  corresponds to the encoding of the number $0$.  

Next, for all $1 \leq i \leq n$, we define $K_i$ as the following
graph
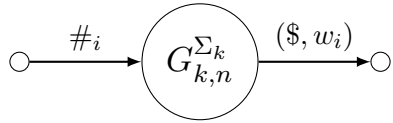
\begin{center}
\begin{tikzpicture} [scale=0.6]
\node(T0) [draw, circle,scale=0.7] at (-4,0) {};
\node(T) [draw, circle,scale=1.2] at (0,0) {$G_{k,n}^{\Sigma_k}$};
\node(S) [draw, circle,scale=0.7] at (4,0) {};
\draw [->,>=latex, thick]
    (T0) -- (T) node[midway,above]{$\#_i$};
\draw [->,>=latex, thick]
    (T) -- (S) node[midway,above]{$(\$,w_i)$};
\end{tikzpicture}
\end{center}
Recall that we want to define $K_i$ in such a way that the trace of
$\pi$ on $K_i$ is the encoding of the contents of the cell with number
$\mathit{in}$ in the initial configuration (that it, it contains the letter
$w_i$ and the head is not scanning the cell since $i \neq 0$). Recall
that the encoding of such a cell (and its content) is given by
\[
c(i) (\$,w_i),
\] 
where $c(i)$ is the $k$-counter encoding the number $i$. We will use
the formula $\mathit{num}_{k,n}^i$ given by Lemma~\ref{lem:pass} to
express that the passage of $\pi$ through the subgraph
$G_{k,n}^{\Sigma_k}$ of $K_i$ corresponds to the $k$-counter encoding
$i$. 

Finally we define the graph $K$ as the graph
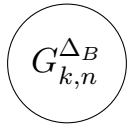
\begin{center}
\begin{tikzpicture} [scale=0.6]
\node(T) [draw, circle,scale=1.2] at (0,0) {$G_{k,n}^{\Delta_B}$};
\end{tikzpicture}
\end{center}
where $\Delta_B$ is the one-letter alphabet containing the blank
symbol $B$. 
Recall that the trace of $\pi$ on the graph $K$ will encode the
contents of the cells with number $\geq n$ in the initial configuration
(that is, the fact that those cells contain the blank symbol and are
not scanned by the head). Since the encoding of such a cell
with number $i$ 
is given by
\[
c(i) (\$,B)
\] 
(where $c(i)$ is the $k$-counter encoding $i$),  the label of the trace of
$\pi$ on $K$ must contain the word
\[
c(n+1) (B,\$) \dots c(j) (B,\$),
\]
where $j=\tow{k}{n}-1$. That is, the label of the trace of $\pi$ on $K$
is the unique $(k,f_0 n,n+1)$-description over the alphabet $\Delta_B$. We
will express that the passage of $\pi$ through the graph
$G_{k,n}^{\Delta_B}$ corresponds to the $(k,f_0 n,n+1)$-description over the
alphabet $\Delta_B$ using the formula $\phi_{k,n,n+1}^{\Delta_B}$ provided
by Lemma~\ref{lem:pass}.

We are now ready to define the graph $I_w$ which is obtained by assembling the
graphs previously introduced in the following way.

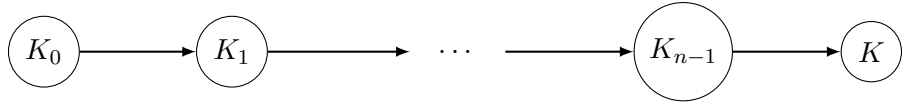
\begin{center}
\begin{tikzpicture} [scale=0.5]

\node(0) [draw, circle,scale=1] at (0,0) {$K_0$};
\node(1) [draw, circle,scale=1] at (5,0) {$K_1$};
\draw [->,>=latex, thick]
    (0) -- (1);

\node(2) at (10,0) {};
\draw [->,>=latex, thick]
    (1) -- (2);

\node(d) at (11,0) {$\dots$};
\node(n1) at (12,0) {};
\node(n) [draw, circle,scale=1] at (17,0) {$K_{n-1}$};

\draw [->,>=latex, thick]
    (n1) -- (n);
\node(i) [draw, circle,scale=1] at (22,0) {$K$};
\draw [->,>=latex, thick]
    (n) -- (i);

\end{tikzpicture}
\end{center}

Each edge between two graphs in the picture above is an edge from the ``left-most'' node
of the first graph to the ``right-most'' node of the second graph. Finally the graph $G_w$ is the graph obtained by
considering the union of the graph $I_w$ and $H$ and 
adding an edge from the  final node of $K$ to the
initial node of $H$. 

Now  that we have defined the graph $G_w$, we are ready to define the
formula $\psi$. The formula $\psi(\pi)$ is obtained as the conjunction of the following
formulas.

\begin{enumerate}[label=\({\Alph*}]
\item  First we need to express that the path $\pi$ starts with the
  edge 
  with label $\mathit{in}$.
\item  We need to express that eventually in a configuration, the
  machine reaches the final state $q_f$. 
\item We also have to express that each passage of the path $\pi$ from the
  initial node of $G_{k,n}^{\Sigma_k}$ to the final node of 
$G_{k,n}^{\Sigma_k}$ in the graph $H$ corresponds to the encoding
  of a $(k+1,f_0 n)$-description. 
\item We have to express that for all $i < n$, the trace of $\pi$ on the
  subgraph $G_{k,n}^{\Sigma_k}$ of the graph $K_i$ corresponds to the
  $k$-counter encoding $i$.
\item We need to express that the trace of the path $\pi$ on the
  subgraph $G_{k+1,n}^{\Delta_B}$ of $K$ is the unique
  $(k+1,f_0 n,n+1)$-description over the alphabet $\Delta_B$. 
\item  Finally we need to express how we move from one configuration of the
  tape to the next one. 
\end{enumerate}
Cases~(A) and~(B) are straightforward. Cases~(C),~(D) and~(E) are
similar and we only give details for case~(C) and case~(F). By Lemma~\ref{lem:pass}, case~(C) means that
\begin{equation} \label{eq:pist} 
 \text{if $\pi_{s,t}$ is the subpath of $\pi$ corresponding to such a
  passage, then $\phi_{k+1,n,0}^{\Delta} (\pi_{s,t})$ holds. }
\end{equation}

The node $s$ is a node satisfying $i_{k,n}^{\Sigma_k}$, while $t$ is the
``closest'' node to $s$ with an incoming edge with label $f_{k+1,n}^\Delta$. This is expressed
by the following
formula   
$\mathit{IF}^\Delta_{k+1,n} (s,t,\pi)$ defined by
\begin{equation} \label{eq:if1}
i_{k+1,n}^{\Sigma_k} (s) \wedge E_{f_{k+1,n}^{\Sigma_{k=1}}} (t-1,t) \wedge \neg \exists r^\pi \;
[s < r <t \wedge E_{f_{k+1,n}^{\Sigma_{k+1}}} (r-1,r)].
\end{equation}
It follows from the definitions of $\mathit{IF}^\Delta (s,t,\pi)$ and
the graph $G_{k+1,n}^\Delta$ that
\begin{align} 
G_{k+1,n}^\Delta \vDash \mathit{IF}^\Delta_{k+1,n} (s,t,\pi)  \; \;
\text{iff}  & \; \; \text{$s$
is the first node of a passage of $\pi$ through
 $G_{k+1,n}^{\Delta}$} \notag  \\ 
 & \; \; \text{and $t$ is the last node of that same passage.}  \label{eq:if}
\end{align}
Hence,~\eqref{eq:pist} is equivalent to the fact that if $
\mathit{IF}^\Delta_{k+1,n} (s,t,\pi) $ holds, so does the formula
$\phi_{k+1,n,0}^{\Delta} (\pi_{s,t})$. Therefore 
the formula 
\[
\forall s^\pi,t^\pi \; [IF^{\Delta} (s,t,\pi) \to \phi_{k+1,n,0}^{\Delta} (\pi_{s,t})]
\]
expresses case~(C). 

Finally we treat the most difficult case which is case~(F). We need to
express how we move from one configuration of the
  tape to the next one. Recall that the trace of $\pi$ on the graph
  $H$ will contain the encoding of the sequence $C_1 C_2 \dots$ of
  the run, where $C_0 C_1 \dots$ is the full run of the machine on the
  input $w$. 

Let $\pi_{s,t}$ be the subpath of $\pi$ corresponding to the encoding
of a configuration $C_i$ and let $\pi_{s',t'}$ be the subpath of $\pi$
corresponding to the configuration $C_{i+1}$. We need to express how
to move from the configuration $C_i$ to the configuration
$C_{i+1}$. Suppose that in the configuration $C_i$, the current state
is $q$, and the head is scanning the cell $c$ containing the letter
$u$. Suppose also that $\delta(q,u)=(q',v,R)$ (we can treat
similarly the case where the head moves to
the left). In order to keep our formulas simpler, we use a slightly
different definition of a run of a Turing machine, but it would be
clear that the notion of run that we use here, can be simulated by a
usual Turing machine. Here, we assume that if the machine scans a cell
$c$ with content $u$ and $\delta(q,u)=(q',v,R)$, then in the next
state, the machine scans the successor $c'$ of $c$, the content of
$c'$ is $v$, while the content of $c$ is $u$ (in the usual definition,
the content of $c'$ is unchanged, while the content of $c$ is $v$). 

Let $\pi_{r,s}$ be the subpath of $\pi$ corresponding to the encoding
of the cell $c$ in the configuration $C_i$. Let $\pi_{r',s'}$ be the
encoding of an arbitrary cell $c'$ in the configuration $C_{i+1}$. If
$c'$ is the successor of the cell $c$, then the head should scan the
cell $c'$ and the content of $c'$ should be the letter $v$. We express
this by the formula $\mathit{change}^R_{(q,a,q',b)}(r,s,r',s',\pi)$ defined by 
\[
\mathit{succ}_{k,n} (\pi_{s,r},\pi_{r',s'}) \wedge (q,u) (s+2) \to
(q',v) (s'+2).
\]
Recall that by Lemma~\ref{lem:pass}, $\mathit{succ}_{k,n}
(\pi_{r,s},\pi_{r',s'})$ is the formula expressing that the
$k$-counter associated with $\pi_{r',s'}$ is the successor of the
$k$-counter associated with $\pi_{r,s}$. 

If $c'$ is not the successor of the cell $c$, then the head is not
scanning the cell $c'$ and its content remains unchanged. If
$\pi_{x,y}$ is the subpath of $\pi$ corresponding to the content of
the cell $c'$ in the configuration $C_i$, this is expressed by the
following formula $\mathit{stay}^R_{(q,u,q',v)}(r,s,x,y,r',s',\pi)$
defined by
\[
\neg \mathit{succ}_{k,n} (\pi_{r,s},\pi_{r',s'}) \wedge (q,u) (s+2)
\wedge \mathit{eq}_{k,n} (\pi_{x,y},\pi_{r',s'}) \wedge (q'',u_0) (x,y) \to (\$,u_0) (s+2)
\]
where $q'' \in Q \cup \{ \$\}$. 
Recall that by Lemma~\ref{lem:pass}, $\mathit{eq}_{k,n}
(\pi_{x,y},\pi_{r',s'})$ expresses that the $k$-counters associated
with $\pi_{x,y}$ and $\pi_{r',s'}$ are the same. 

Now we need to express that the paths $\pi_{r,s}$, $\pi_{s',s'}$ and
$\pi_{x,y}$ correspond to the encodings of $k$-counters. By
Lemma~\ref{lem:pass}, this means that those paths correspond to
passages of $\pi$ through the graph $G_{k,n}^{\Sigma_k}$. Similarly
to~\eqref{eq:if1}, we introduce a formula
$\mathit{IF}^{\Sigma_k}_{k,n}(r,s,\pi)$ defined by
\[
i_{k,n}^{\Sigma_k} (r) \wedge E_{f_{k,n}^{\Sigma_k}} (s-1,s) \wedge \neg \exists t^\pi \;
[r < t <s \wedge E_{f_{k,n}^{\Sigma_k}} (t-1,t)].
\]
This formula expresses that the path $\pi_{r,s}$ corresponds to the
encoding of a $k$-counter. 

Next we also need a formula to assert that
the paths $\pi_{r,s}$ and $\pi_{r',s'}$ appear in the encodings of
successive configurations (and similarly, that the paths $\pi_{x,y}$ and $\pi_{r',s'}$ appear in the encodings of
successive configurations). Since the encoding of a configuration
starts with the unique edge with label $i_{k+1,n}^\Delta$ (and that
edge only occurs at the beginning of the encoding of a configuration), this is
equivalent to say that there is a unique edge between $s$ and $r'$
with label $i_{k+1,n}^\Delta$. This is expressed by the formula
$\mathit{config}(s,r',\pi)$ defined by
\[
\exists t^\pi \; [(s < t < r') \wedge i_{k+1,n}^\Delta (t) \wedge
\neg \exists (t')^\pi \; [(s < t' < r') \wedge i_{k+1,n}^\Delta (t')  \wedge
(t' \neq t)] ].
\]

We are now ready to define 
$\theta^R_{q,u,q',v} (\pi) $ as the following formula 
\begin{align}
 && \forall r^\pi,s^\pi,(r')^\pi,(s')^\pi,x^\pi,y^\pi \; [\mathit{IF}_{k,n}^{\Sigma_k}(r,s,\pi)
\wedge \mathit{IF}_{k,n}^{\Sigma_k}(r',s',\pi)  \wedge
\mathit{IF}_{k,n}^{\Sigma_k}(x,y,\pi)   \label{eq:1}\\ &&
\wedge \mathit{config}(s,r',\pi) \wedge
\mathit{config}(y,r'\pi)  \label{eq:2}\\  && \to
\mathit{change}^R_{(q,u,q',v)}(s,r,s',r',\pi) \wedge
\mathit{stay}^R_{(q,u,q',v)}(r,s,x,y,r',s',\pi)]. \label{eq:3}
\end{align}
It expresses the following. Suppose that $\pi_{r,s}$, $\pi_{r,s}$
and $\pi_{r',s'} $ are $k$-counters encoding the numbers of three
cells (this corresponds to~\eqref{eq:1}). Suppose that $\pi_{r,s}$ and
$\pi_{x,y}$ correspond to cells occurring in the same configuration $C$
and that the cell corresponding to $\pi_{r',s'}$ occurs in the next
configuration  $C'$ (this is expressed by~\eqref{eq:2}). Then, if we
``apply'' the transition $\delta(q,u)=(q',v,R)$ to move from $C$ to
$C'$, we  move the head to the right and update the content of the
cell being scanned (as expressed by the formula
$\mathit{change}^R_{(q,u,q',v)}(r,s,r',s',\pi)$) and we leave the
other cells unchanged (as expressed by the formula
$\mathit{stay}^R_{(q,u,q',v)}(r,s,x,y,r',s',\pi)$).  

We define now the formula $\theta_R(\pi)$ as the formula
\[
\bigwedge \{ \theta^R_{q,u,q',v} (\pi) : \delta(q,u)=(q',v,R)\}.
\]
This formula expresses how we move from one configuration to another,
when the head moves to the right. Similarly, we can define a formula
$\theta_L(\pi)$ expressing how we move from one configuration to another,
when the head moves to the left.

This finishes the proof of Theorem \ref{theo:main}. \qed


\medskip

As a corollary to the proof of Theorem \ref{theo:main}, we obtain that
data complexity is non-elementary even for simple WL formulas that
talk about a single path in a graph database.

\begin{corollary} The evaluation problem for WL over graph databases
is non-elementary in data complexity, even if restricted to Boolean WL
formulas of the form $\exists \pi \psi$, where $\psi$ uses no path
quantification and contains no position variable of sort different than 
$\pi$. \end{corollary}

\section{Register Logic} 
\label{sec:rl}

We saw in the previous section that WL is impractical due to its 
very high data complexity. In this section, we start by recalling the
notion of regular expressions with memory (REM) and their basic
results from \cite{LV}. In our view, this logic is rather limited in
terms of expressive power.  For instance, the query (Q) from the
introduction cannot be expressed in REM. We then introduce an
extension of REM, called regular logic (RL), that remedies this
limitation in expressive power (in fact, it can express many natural
examples of queries expressible in WL, e.g., those given in \cite{WL})
while retaining elementary complexity of query evaluation. Finally, we
study which fragments of RL are well-behaved
for database applications.

\subsection{Regular expressions with memory}

REMs define pairs of nodes in data graphs that are linked by a path
that satisfies a constraint in the way in which the topology
interacts with the underlying data.  
REMs allow us to remember data values and use them later. 
Data values
are stored in $k$ registers 
$r_1,\dots,r_k$. At any point we can compare a data value with one
previously stored in the registers. As an example, consider the REM
$\downarrow \! r. a^+[r^=]$. This can be read as follows: Store the current
data value  in register $r$ (represented by the expression $\downarrow \! r$), and then check that after reading a word in $a^+$
we see the same data value again (condition $[r^=]$). We formally
define REM next. 

Let $r_1,\dots,r_k$ be registers. The set of {\em conditions} $c$ over
$\{r_1,\dots,r_k\}$ is recursively defined as: $c \; := \; r_i^= \, \mid \, c
\wedge c \, \mid \, \neg c$, for $1 \leq i \leq k$. Assume that $\V_\bot$ 
is the extension of the set $\D$ of data values 
with a new symbol $\bot$. 
Satisfaction of conditions is defined with respect to a value $d \in
\V$ (the data value that is currently being scanned) and a tuple $\tau =
(d_1,\dots,d_k) \in \V_\bot^k$ (the data values stored in the registers,
assuming that $d_i = \bot$ represents the fact that register $r_i$ has no value
assigned) as
follows (Boolean combinations omitted): $(d,\tau) \models r_i^=$ iff $d
= d_i$.

\begin{definition}[REMs]
The class of REMs over $\Sigma$ and $\{r_1,\dots,r_k\}$ is defined by the grammar:
$$e \ := \ \varepsilon \: \mid \: a \: \mid \: e \cup e \: \mid \: 
e \cdot e \: \mid \: e^+
\: \mid \: e[c] \: \mid \:\, \downarrow \!\bar r. e$$
where $a$ ranges over symbols in $\Sigma$, $c$ over conditions over
$\{r_1,\dots,r_k\}$, and $\bar r$ over tuples of elements in
$\{r_1,\dots,r_k\}$. \qed
\end{definition} 

That is, REM extends the class of regular expressions $e$ -- which 
is a popular mechanism for specifying topological properties of paths in graph
databases (see, e.g., \cite{Wood,Bar13}) -- 
with expressions of the form $e[c]$, for $c$ a condition, and
$\downarrow \! \bar r. e$, for $\bar r$ a tuple of registers -- that
define how such topology interacts with the data.

\paragraph{{\bf Semantics:}} 
To define the evaluation $e(G)$ of an REM $e$ over a data graph $G = (V,E,\kappa)$, 
we use a relation $\sem{e}_G$ that consists of tuples of the form 
$(u,\lambda,\rho,v,\lambda')$, for $u,v$ nodes in $V$, $\rho$ a path in $G$ from $u$ to $v$,  
and $\lambda,\lambda'$ two $k$-tuples over $\V_\bot$. 
The intuition is the following: the tuple $(u,\lambda,\rho,v,\lambda')$ belongs to
$\sem{e}_G$ if and only if the data and topology of $\rho$ can be parsed
according to $e$, with $\lambda$ being the initial assignment of the
registers, in such a way that the final assignment is $\lambda'$. 
We then define $e(G)$ as the pairs $(u,v)$ of nodes in $G$ such that 
$(u,\bot^k,\rho,v,\lambda) \in \sem{e}_G$, for some path $\rho$ in $G$ from $u$ to $v$ and 
$k$-tuple
$\lambda$ over $\V_\bot$.  

We inductively define relation
$\sem{e}_G$ below. 
We assume that $\lambda_{\bar r = d}$, for $d \in \D$, is the tuple
obtained from $\lambda$ by setting all registers in $\bar r$ to be
$d$. Also, if $\rho_1 = v_1 a_1 v_2
\cdots v_{k-1} a_{k-1} v_k$ and $\rho_2 = v_k a_k v_{k+1} \cdots v_{n-1} a_{n-1} v_n$ are paths, 
then: $$\rho_1 \rho_2 \ := \ 
v_1 a_1 v_2 \cdots v_{k-1} a_{k-1} v_k a_k v_{k+1} \cdots v_{n-1} a_{n-1} v_n.$$ 
Then we define: 

\begin{itemize}

\item 
$\sem{\varepsilon}_G = \{(u,\lambda,\rho,u,\lambda) : u \in V, \, \rho = u, \, \lambda \in \D_\bot^k\}$. 

\item 
$\sem{a}_G = \{(u,\lambda,\rho,v,\lambda) : \rho = u a v, \, \lambda \in \D_\bot^k\}$. 

\item 
$\sem{e_1 \cup e_2}_G = \sem{e_1}_G \cup \sem{e_2}_G$. 

\item $\sem{e_1 \cdot e_2}_G = \sem{e_1}_G \circ \sem{e_2}_G$, where 
$\sem{e_1}_G \circ \sem{e_2}_G$ is the set of tuples $(u,\lambda,\rho,v,\lambda')$ 
such that 
$(u,\lambda,\rho_1,w,\lambda'') \in \sem{e_1}_G$ and $(w,\lambda'',\rho_2,v,\lambda') \in \sem{e_2}_G$, 
for some $w \in V$, $k$-tuple
$\lambda''$ over $\V_\bot$, and paths $\rho_1,\rho_2$ such that 
$\rho = \rho_1 \rho_2$. 

\item 
$\sem{e^+}_G = \sem{e}_G \cup (\sem{e}_G \circ \sem{e}_G) \cup (\sem{e}_G \circ \sem{e}_G \circ \sem{e}_G) \dots$

\item $\sem{e[c]}_G =$ $\{(u,\lambda,\rho,v,\lambda') \in \sem{e}_G :$ $(\kappa(v),\lambda') \models c\}$. 

\item $\sem{\downarrow \!\bar r. e}_G =$ $\{(u,\lambda,\rho,v,\lambda') : (u,\lambda_{\bar r = \kappa(u)},\rho,v,\lambda') \in \sem{e}_G\}$.  

\end{itemize} 
For each REM $e$, we will use the shorthand notation $e^*$ to denote
$\varepsilon \cup e^+$.

\begin{example} The REM $\Sigma^* \cdot (\downarrow \!\! r. \Sigma^+
  [r^=]) \cdot \Sigma^*$ defines the pairs of nodes in data graphs 
that are linked by
  a path in which two nodes have the same data value. The REM
  $\downarrow \!r. (a[\neg r^=])^+$ defines the pairs of nodes that
  are linked by a path $\rho$ with label in $a^+$, such that the data
  value of the first node in the path is different from the data value
  of all other nodes in $\rho$.  \hfill $\Box$ \end{example}

The problem {\sc Eval}(REM) is, 
given a data graph $G = (V,E,\kappa)$, a pair $(v_1,v_2)$ of nodes in
$V$, and an REM $e$, is $(v_1,v_2) \in e(G)$?  
The data complexity of the problem refers again to
the case when $e$ is considered to be fixed.  REMs are tractable in
data complexity and have no worse combined complexity than FO over
relational databases:

\begin{proposition}[\cite{LV}] 
{\sc Eval}{\em (REM)} is {\sc Pspace}-complete, 
and {\sc Nlogspace}-complete in data complexity.  
\end{proposition}

\subsection{Register logic}

REM is well-behaved in terms of the complexity of evaluation, but its
expressive power is rather rudimentary for expressing several
data/topology properties of interest in data graphs. As an example,
the query (Q) from the introduction -- which can be easily expressed
in WL -- cannot be expressed as an REM (we actually prove a stronger
result later). The main shortcomings of REM in terms of its expressive
power are its inability to (i) compare data values in different paths
and (ii) express branching properties of the data. 
 
In this section, we propose register logic (RL) as a natural extension
of REM that makes up for this lack of expressiveness.  We borrow ideas
from the logic CRPQ$^\neg$, presented in \cite{BLWW}, that closes the
class of {\em regular path queries} \cite{CMW} under Boolean
combinations and existential {\em node} and {\em path} quantification.
In the case of RL we start with REMs and close them not only under
Boolean combinations and node and path quantification -- which allow
to express arbitrary patterns over the data -- but also under {\em
  register assignment} quantification -- which permits comparing data
values in different paths. We also prove that the combined complexity of the
evaluation problem for RL is elementary (\expspace), and, thus, that in
this regard RL is in stark contrast to WL. 

To define RL we assume the existence of countably infinite sets of
{\em node}, {\em path} and {\em register assignment variables}. Node
variables are denoted $x,y,z,\dots$, path variables are denoted
$\pi,\pi',\pi_1 ,\pi_2,\dots$, and register assignment variables are denoted
$\nu,\nu_1,\nu_2,\dots$

\begin{definition}[Register logic (RL)] 
We define the class of RL formulas
$\phi$ over alphabet
$\Sigma$ and $\{r_1,\dots,r_k\}$ using the following
grammar:
\begin{eqnarray*}
{\tt atom} \ & := & \ x = y \, \mid \, \pi = \pi' \, \mid \, \nu
= \nu' \, \mid \, \nu = \bar \bot \, \mid \, 
(x,\pi,y) \, \mid \, e(\pi,\nu_1,\nu_2) \\
\phi \ & := & \ {\tt atom} \, \mid \, \neg \phi \, \mid \, \phi \vee
\phi \, \mid \, \exists x \phi \, \mid \, \exists \pi \phi \, \mid \,
\exists \nu \phi
\end{eqnarray*}
Here $x,y$ are node variables, $\pi,\pi'$ are path variables, 
$\nu,\nu'$ are register assignment variables, and $e$ is an REM over
$\Sigma$ and $\{r_1,\dots,r_k\}$. \qed
\end{definition} 

Intuitively, $\nu = \bar \bot$ holds iff $\nu$ is the empty register
assignment, $(x,\pi,y)$ checks that $\pi$ is a path from $x$ to $y$,
and $e(\pi,\nu,\nu')$ checks that $\pi$ can be parsed according to $e$
starting from register assignment $\nu$ and finishing in register
assignment $\nu'$. The quantifier $\exists \nu$ is to be read ``there
exists an assignment of data values in the data graph to the
registers''. 

Let $G = (V,E,\kappa)$ be a data graph over $\Sigma$ and $\phi$ a RL
formula over $\Sigma$ and $\{r_1,\dots,r_k\}$. Assume that $D$ is the
set of data values that are mentioned in $G$, i.e., $D = \{\kappa(v)
: v \in V\}$. 
An {\em assignment}
$\alpha$ for $\phi$ over $G$ is a mapping that assigns (i) a node in
$V$ to each free node variable $x$ in $\phi$, (ii) a path $\rho$ in
$G$ to each free path variable $\pi$ in $\phi$, and (iii) a tuple
$\lambda$ in $\D_\bot^k$ to each register variable $\nu$
that appears free in $\phi$.  That is, for safety reasons we assume
that $\alpha(\nu)$ can only contain data values that appear in the
underlying data graph. This represents no restriction for the
expressiveness of the logic.
 
We inductively define $(G,\alpha) \models \phi$, for $G$ a data graph,
$\phi$ an RL formula, and $\alpha$ an assignment for $\phi$ over $G$,
as follows (we omit equality atoms and Boolean combinations since they
are standard):
 
\begin{itemize}

\item $(G,\alpha) \models \nu = \bar \bot$ iff $\alpha(\nu) =
\bot^k$. 

\item $(G,\alpha) \models (x,\pi,y)$ iff $\alpha(\pi)$ is a path from
$\alpha(x)$ to $\alpha(y)$ in $G$. 

\item $(G,\alpha) \models e(\pi,\nu,\nu')$ iff  
$(u,\alpha(\nu),\alpha(\pi),v,\alpha(\nu')) \in \sem{e}_G$, assuming 
$\alpha(\pi)$ goes from node $u$ to $v$. 

\item $(G,\alpha) \models \exists x \phi$ iff there is node $v \in V$
such that $(G,\alpha[x \to v]) \models \phi$. 

\item $(G,\alpha) \models \exists \pi \phi$ iff there is path $\rho$ in $G$
such that $(G,\alpha[\pi \to \rho]) \models \phi$.

\item $(G,\alpha) \models \exists \nu \phi$ iff there is tuple
$\lambda$ in
$\D_\bot^k$ such that $(G,\alpha[\nu \to \lambda]) \models
\phi$. 

\end{itemize} 

Thus, each REM $e$ is expressible in RL using the
formula: $$\exists \pi \exists \nu \exists \nu' \, ( \, \nu = \bar \bot \,
\wedge \, 
e(\pi,\nu,\nu') \,).$$  

\begin{example} \label{exa:trial1} 
Recall query (Q) from the introduction:
{\em Find pairs of nodes $x$ and $y$ in a graph database, such that there is
a node $z$ and a path $\pi$ from $x$ to $y$ in which each node 
is connected to $z$}.
This query can be expressed in RL over
$\Sigma = \{a\}$ and a
single register $r$ as follows:
$$\exists \pi \,\big(\,(x,\pi,y) \wedge \exists z \forall \nu
(e_1(\pi,\nu,\nu) \rightarrow \exists z' \exists \pi'
((z',\pi',z) \wedge e_2(\pi',\nu,\nu)))\,\big),$$
where $e_1 := a^* [r^=] \cdot a^*$ is the REM that checks whether the node (i.e. data) stored in
register $r$ appears in a path, and $e_2 := \varepsilon [r^=] \cdot a^*$
is the REM that checks if the first node of a path is the one that is
stored in register $r$. 

In fact, this formula defines the pairs of nodes $x$ and $y$ such that
there exists a path $\pi$ that goes from $x$ to $y$ and a node $z$ for which the
following holds: for every register
value $\nu$ (i.e., for every node $\nu$) such that
$e_1(\pi,\nu,\nu)$ (i.e. node $\nu$ is in $\pi$), it
is the case that there is a path $\pi'$ from some node $z'$ to $z$
such that $e_2(\pi',\nu,\nu)$ (i.e., $z' = \nu$ and $\pi'$
connects $\nu$ to $z$). Notice that this uses the fact that
the underlying data model is that of graph databases, in which each
node is uniquely identified by its data value. 
\boxtheorem
\end{example} 

The limitations in expressive power 
of RL have also been independently recognized by Libkin, Martens and Vrgoc
\cite{LMV}. In order to allow for interesting data value comparisons
while retaining reasonable complexity of evaluation, they propose to
use query languages based on the XML language XPath. These languages are not
comparable in terms of expressive power to the ones we study here. 

\medskip 

\noindent
{\bf Complexity of evaluation for RL:} 
The evaluation problem for RL, denoted {\sc Eval}(RL), is as
follows: Given a data graph 
$G$, an RL formula $\phi$, and an assignment $\alpha$ for $\phi$ over
$G$, is it the case that $(G,\alpha) \models \phi$?  
As before, we denote by {\sc Eval}(RL,$\phi$) the evaluation problem
for the fixed RL formula $\phi$.  

We show next that, unlike WL, 
register logic RL can be evaluated in elementary time,
and, actually, with only one exponential jump over the complexity of
evaluation of REMs:

\begin{theorem} \label{theo:rl-cc} 
{\sc Eval}{\em (RL)} is {\sc Expspace}-complete. The lower bound holds
even if the input is restricted to graph databases. 
\end{theorem}

\begin{proof}

We start by proving the upper bound, that is, 
{\sc Eval}{\em (RL)} is in {\sc Expspace}. 
The structure of the proof is quite similar to the one that proves
 that CRPQ$^\neg$ queries 
can be evaluated in {\sc Pspace} in combined complexity
\cite{BLWW}. The difference is that now 
we have to accommodate 
the extra expressive power of RL, 
that
allows to express properties of register values
and check acceptance of data walks by REMs.

Let $G = (V,E,\lambda)$ be a $\Sigma$-labeled data graph and $\phi$ a RL formula over $\Sigma$ and
$\{x_1,\dots,x_k\}$. Let us denote by $D = \{\lambda(v) : v \in
V\}$ and define $D_\bot$ to be $D \cup \{\bot\}$. 
Further, let $\alpha$ be an assignment for $\phi$ over $G$. %
We define 
\begin{itemize}
\item $\bar v = (v_1,\dots,v_{k_1})$ as a tuple of nodes in $G$ such
that 
\[
\{ v_1,\dots, v_{k_1}\} = \{ \alpha(x) : x \text{ is a free node
  variable}\},
\]
\item $\bar \rho = (\rho_1,\dots,\rho_{k_2})$ as a tuple of paths in
$G$ such that  
\[
\{ \rho_1,\dots, \rho_{k_2}\} =  \{ \alpha(\rho) : \rho
\text{ is a free path 
  variable}\},\]  
\item  
$\bar \lambda = (\lambda_1,\dots,\lambda_{k_3})$ as a tuple of
register values for $\{x_1,\dots,x_k\}$ over $G$ (i.e., tuples in $(D \cup
\{\bot\})^k$) such that $\{ \lambda_1,\dots, \lambda_{k_3}\} =$ $ \{ \alpha(\lambda) : \rho
\text{ is a free register  
  variable}\}$. 
\end{itemize}
Further, 
assume that $e_1,\dots,e_m$ are all the REMs
mentioned in $\phi$. Our goal is to define an \expspace\ procedure that
checks whether $(G,\alpha) \vDash \phi$. 
In order to
do that, we first have to introduce some new terminology.

Let $\tau$ be a first-order (FO) vocabulary 
$$\langle {\tt Nodes}, {\tt Paths}, {\tt Registers}, 
{\tt Endpoints}, e_1,\dots,e_m, \bar \bot \rangle,$$ where 
(a) ${\tt Nodes}$, ${\tt Paths}$ and ${\tt Registers}$  
 are unary relation symbols, (b) ${\tt Endpoints}$ 
and $e_i$
($1 \leq i \leq m$)
are
ternary relation symbols, and (c) $\bar \bot$ is a constant.  
We define, from $G$, an FO structure $\M_G$
over $\tau$ as follows: The domain of $\M_G$ is the disjoint union of
$V$, all the paths that belong to $G$, and all $k$-tuples over  
$D_\bot$. (Notice that each node in
$V$ is also a path in $G$, but here we consider them to be different
objects. That is, each $v \in V$ appears separately as a node and as a
path in the domain of $\M_G$). The constant $\bar \bot$ is interpreted
in $\M_G$ as the tuple $\bot^k$. 
The interpretation of ${\tt Nodes}$ in $\M_G$
contains all those elements of the domain that are nodes. The
interpretation of ${\tt Paths}$ in $\M_G$ contains all those elements of the
domain that are paths.
The
interpretation of ${\tt Registers}$ in $\M_G$ contains all those elements of the
domain that are $k$-tuples over  
$D_\bot$.
The interpretation of the ternary relation
${\tt Endpoints}$ contains all tuples $(v,\rho,v')$ such that $\rho$ is a
path in $G$ from node $v$ to node $v'$. Finally, the interpretation of
the symbol $e_i$ ($1 \leq i \leq m$) contains all tuples
$(\lambda,\rho,\lambda')$ such that $\rho$ is a path in $G$,
$\lambda,\lambda'$ are $k$-tuples over $D_\bot$, and 
$e_i(\rho,\lambda,\lambda')$. 


Let $\phi_\tau$ be the \FO\ formula over vocabulary $\tau$ obtained
from $\phi$ by simultaneously replacing (1) each subformula of the
form $\exists x \theta$ (for $x$ a node variable) with $\exists x
({\tt Nodes}(x) \wedge \theta)$, (2) each subformula of the form $\exists
\pi \theta$ (for $\pi$ a path variable) with $\exists \pi ({\tt Paths}(\pi)
\wedge \theta)$, (3) each subformula of the form $\exists
\nu \theta$ (for $\nu$ a register variable) with $\exists \nu ({\tt Registers}(\nu)
\wedge \theta)$, and (4) each atomic formula of the form $(x,\pi,y)$
with ${\tt 
Endpoints}(x,\pi,y)$. 

Clearly, $G,\alpha \vDash \phi$ iff $\M_G,\alpha \vDash \phi_\tau$. 

Of course, $\M_G$ cannot be effectively constructed from $G$ since the
set of paths in $G$ is potentially infinite, and, thus, $\M_G$ is also
potentially infinite. However, it is possible to prove that there
exists a finite structure $\M'_{G,\bar \rho}$ such
that $G,\alpha \vDash \phi$ iff $\M'_{G,\bar \rho},\alpha \vDash \phi_\tau$. 
 We show how to define $\M'_{G,\bar
\rho}$ next.

\newcommand{\E}{{\mathcal E}}

Assume that the {\em quantifier rank} of $\phi_\tau$ is $k \geq 0$,
where the quantifier rank of an \FO\ formula $\theta$ is the
depth of nested quantification in $\theta$.  Let $\E \subseteq
\{e_1,\dots,e_m\} \times D_\bot^k \times D_\bot^k$. A path $\rho$ in $G$ {\em satisfies} $\E$
if the following holds: For each triple $(e_i,\lambda,\lambda') \in 
\{e_1,\dots,e_m\} \times D_\bot^k \times D_\bot^k$, it is the case that $G,\alpha \vDash e_i(\rho,\lambda, \lambda')$ iff
$(e_i,\lambda,\lambda') \in \E$. (Notice that for each path in $G$
there is one, and only one, subset $\E$ of $\{e_1,\dots,e_m\} \times
D_\bot^k \times D_\bot^k$ that it satisfies.)
For each pair $(v,v')$ of nodes in $V$, and for every $\E \subseteq
\{e_1,\dots,e_m\} \times D_\bot^k \times D_\bot^k$, 
let $c_{\E,v,v'} \geq 0$ be the minimum between $k +
|\bar \rho|$ and the number of paths in $G$ that go from $v$ to $v'$
and satisfy $\E$. We arbitrarily pick, for each pair $(v,v')$ of nodes
in $V$ and for each $\E \subseteq \{e_1,\dots,e_m\} \times D_\bot^k 
\times D_\bot^k$, $c_{\E,v,v'}$ distinct paths
$\rho^1_{\E,v,v'},\dots,\rho^{c_{\E,v,v'}}_{\E,v,v'}$ from $v$ to $v'$
that satisfy $\E$.

We define the structure $\M'_{G,\bar \rho}$ as follows: Its
domain contains all the nodes of $V$, each path $\rho$ that belongs to
the tuple $\bar \rho$, every path of the form $\rho^i_{\E,v,v'}$,
where $\E \subseteq
\{e_1,\dots,e_m\} \times D_\bot^k \times D_\bot^k$ ($v,v' \in V$ and $1 \leq i
\leq c_{\E,v,v'}$), and every tuple in $D_\bot^k$.
The constant $\bar \bot$ is interpreted in $\M'_{G,\bar \rho}$ as the
tuple $\bot^k$.  
The interpretation of ${\tt Nodes}$ in $\M'_{G,\bar \rho}$ contains all nodes in the domain. The interpretation
of ${\tt Paths}$ in $\M'_{G,\bar \rho}$ contains all those elements
of the domain that are paths. 
The
interpretation of ${\tt Registers}$ in $\M_{G,\bar \rho}$ 
contains all those elements of the
domain that are $k$-tuples over  
$D_\bot$.
The interpretation of the ternary
relation ${\tt Endpoints}$ contain all tuples of the form $(v,\rho,v')$,
where $v,v' \in V$ and $\rho$ is a path in the domain that goes from
$v$ to $v'$ in $G$. Finally, the interpretation of $e_i$ ($1 \leq i
\leq m$) in $\M'_{G,\bar \rho}$ contains all tuples
$(\lambda,\rho,\lambda')$ such that $\rho$ is a path in the domain,
$\lambda,\lambda'$ are $k$-tuples over $D_\bot$, and 
$e_i(\rho,\lambda,
\lambda')$. 

By using a standard Ehrenfeucht-Fra\"iss\'e argument it is possible to
prove the following: 

\begin{claim}
The structures $(\M_G,\bar v,\bar \rho,\lambda)$ and $(\M'_{G,\bar
\rho},\bar v,\bar \rho,\lambda)$ are indistinguishable by \FO\ sentences of
quantifier rank $\leq k$. 
\end{claim}

\begin{proof} 
We show that the duplicator has a winning strategy in the
  $k$-round Ehrenfeucht-Fra\"iss\'e game played on $(\M_G,\bar v,\bar
  \rho,\bar \lambda)$ and $(\M'_{G,\bar \rho},\bar v,\bar \rho,\bar \lambda)$. The
  duplicator's response to a spoiler move in round $i \leq k$ is
  (inductively) defined as follows (we assume without loss of
  generality that the spoiler never repeats moves, i.e.\ in no round
  does the spoiler choose an element that has already been chosen by
  either player in previous rounds):  
\begin{itemize} 
\item 
If the spoiler's move in round $i$ is a node in either of the two
structures, then the duplicator responds by mimicking the spoiler's
move on the other structure;
\item if the spoiler's move in round $i$ is a $k$-tuple over $D \cup
\{\bot\}$ 
in either of the two
structures, then the duplicator responds by mimicking the spoiler's
move on the other structure;
\item if the spoiler's move in round $i$ is a path $\rho$ in $\bar \rho$ in either
  of the two structures, then again the duplicator responds by
  mimicking the spoiler's move on the other structure;   
\item if the spoiler plays a path $\rho$ from node $v$ to $v'$, in either of
  the two structures, such that $\rho$ satisfies $\E \subseteq
\{e_1,\dots,e_m\} \times D_\bot^k \times D_\bot^k$ and $\rho$ is not a path in $\bar \rho$, then
  the duplicator responds with any path from $v$ to $v'$ in the other
  structure that (1) satisfies $\E$, (2) does not belong to $\bar
  \rho$,  and (3) has not been previously
  chosen in the game. Notice that it is always possible for the
  duplicator to choose such a path, since for each pair of nodes $v,v'
  \in V$ and for each $\E \subseteq
\{e_1,\dots,e_m\} \times D_\bot^k \times D_\bot^k$, 
the number of paths from $v$ to $v'$ that satisfy $\E$ and do not
belong to $\bar \rho$ is the same up to $k$.   
\end{itemize} 

\noindent It is easy to see that duplicator's response defined in this way
always preserves a partial isomorphism between the two
structures. This implies that the duplicator has a winning strategy in the
  $k$-round Ehrenfeucht-Fra\"iss\'e game played on $(\M_G,\bar v,\bar
  \rho,\bar \lambda)$ and $(\M'_{G,\bar \rho},\bar v,\bar \rho,\bar \lambda)$, and, thus,
  by well-known results, 
that the structures are indistinguishable by
  \FO\ sentences of quantifier rank $\leq k$. 
\end{proof}  

The previous claim shows that  $(G,\alpha) \vDash \phi$ iff
$(M'_{G,\bar\rho}, \alpha) \vDash \phi_\tau$. 
Thus, a straightforward approach to
check whether $(G,\alpha) \vDash \phi$ would be to construct
$\M'_{G,\bar \rho}$ and then evaluate $\phi_\tau$ over it. The
problem with this approach is that $\M'_{G,\bar \rho}$ could be
of double exponential size (because there is a double exponential number 
of different subsets $\E$ of
$\{e_1,\dots,e_m\} \times D_\bot^k \times D_\bot^k$), and, thus, impossible to construct in
exponential space. It will be necessary to follow a different approach.

Assume that $\phi_\tau$ is given in prenex normal form,
i.e.\ $\phi_\tau$ is of the form 
\[
Q_1 y_1 \cdots Q_m y_m \, \psi(\bar
x,\bar \pi,\bar \nu,y_1,\dots,y_m),\]
 where each $Q_i$ is either $\exists$ or
$\forall$, each $y_i$ is a node, path or register variable, 
and $\psi$ is quantifier-free (if $\phi_\tau$ is
not in prenex normal form, we can convert it in polynomial time into
an equivalent formula in prenex normal form). We follow a
usual argument to evaluate \FO\ formulas on structures.  The
main problem with this is that some of the elements in $\M'_{G,
\bar \rho}$ are paths and register values, and 
have to be treated as such. Therefore, we define a way of encoding paths (in exponential space)
and register values (in polynomial space).

Clearly, each register value can be codified with a tuple of length
$k \cdot \log_2{(|V|)}$. In order to denote that this tuple 
is the address of a register value (and not, say, of a path), 
we add an extra bit at the beginning of the tuple which is labeled with
a new symbol $r$. Codification of paths requires a bit of extra work. 
Each path $\rho$ is encoded with an {\em address}, that is, a string
that satisfies the following:
\begin{itemize}
\item It starts with a
new symbol $p$, that states that this is the address of a path; 
\item the address continues with the encodings of the two endpoints
  $v$ and $v'$ of
  the path (separated with some delimiter); this part of the address
uses $O(\log_2{|V|})$ space; 
\item then the address encodes the subset $\E$
  of $\{e_1,\dots,e_m\} \times D_\bot^k \times D_\bot^k$ 
that $\rho$ satisfies; this encoding can be easily expressed 
with a string of length $m \times |V|^k \times |V|^k$ over alphabet 
$\{0,1\}$ that flags with a 1 those elements of 
$\{e_1,\dots,e_m\} \times D_\bot^k \times D_\bot^k$ that belong to $\E$;  
\item finally, the address contains an encoding of the integer $i \leq k +
  |\bar \rho|$ such that $\rho = \rho^i_{\E,v,v'}$; this encoding uses
$O(\log_2{(|\phi|+|\bar \rho|)})$ space.  
\end{itemize} 
Clearly, the address of a path defined in this way can be specified
using at most exponential space.  

We show next how the problem of checking whether $(\bar v,\bar
\rho,\bar \lambda)$
belongs to the evaluation of $\phi_\tau$ over $\M'_{G,\bar \rho}$ can
be solved 
in exponential time by an alternating Turing
  machine. This will finish the proof of the theorem, since the class
  of problems that can be solved in exponential time by alternating
  Turing machines coincides with the class of problems that can be
  solved in \expspace.

The alternating machine proceeds as follows. It first replaces in
$\phi_\tau$ each node variable $x$ in $\bar x$ with the encoding of the
corresponding node $v$ of $\bar v$, each path variable
$\pi$ in $\bar \pi$ with the encoding (address) of the corresponding
path $\rho$ in $\bar \rho$, and each register variable $\nu$ in $\bar
\nu$ with the encoding of the corresponding tuple $\lambda$ in $\bar
\lambda$. Then the machine reads the formula
$\phi_\tau$ from left-to-right. Each time it encounters an existential
quantifier $\exists y_i$ it enters an existential state, and each time
it encounters a universal quantifier $\forall y_i$ it enters a
universal state. In each case, the machine ``guesses'' the
interpretation of $y_i$ as the encoding of
a node, a path or a register value 
$c(y_i)$ in the domain. (Since encodings of paths are of exponential size, this alternating
machine requires at least exponential time to work). 
Finally, the machine verifies
that $\psi(\bar v,\bar \rho,\bar \lambda,c(y_1),\dots,c(y_m))$ holds, and if that
is the case it accepts. We show next that the latter can be done in
exponential time. Notice that this implies that the whole process can
be performed in exponential time.

We start with the case of the atomic formulas in $\psi$. In order to
check whether the element assigned to a variable belongs to the
interpretation of ${\tt Nodes}$ in $\M'_{G,\bar \rho}$, we only have
to check that the encoding of this element does not start with a
$p$ or an $r$. In order to check whether the element belongs to the
interpretation of ${\tt Paths}$ (resp., ${\tt Registers}$), 
it is sufficient to check that its encoding
starts with a $p$ (resp., $r$). 
In order to check whether the elements $a,b,c$
assigned to variables $x,\pi,y$, respectively, are such that $(a,b,c)$
belongs to the interpretation of ${\tt Endpoints}$, we only have to check
that $b$ is the encoding of a path, $a$ and $c$ are encodings of
nodes, and that $b$ is a path from $a$ to $c$. 
 Finally, in order to check whether the
elements $(a,b,c)$ assigned to a variable belongs to the interpretation of
$e_i$ ($1 \leq i \leq m$), we only have to check that $a,c$ are
register values (i.e. their encodings start with symbol $r$), that $b$ encodes a path
$\rho$ (i.e. its encoding starts with $p$), and that the bit that corresponds to
tuple $(e_i,a,c)$ in the part of the address $b$ that encodes the
set $\E \subseteq \{e_1,\dots,e_m\} \times D_\bot^k \times
D_\bot^k$ that $\rho$ satisfies is set to 1. 

Thus, the value of the atomic formulas involved in $\psi(\bar v,\bar
\rho,\bar \lambda,c(y_1),\dots,c(y_m))$ can be computed in polynomial
time (in the size of $\psi(\bar v,\bar
\rho,\bar \lambda,c(y_1),\dots,c(y_m))$). But
since $\psi$ is a polysize Boolean combination of atomic formulas,
 the value of $\alpha(\bar v,\bar \rho,\bar \lambda,c(y_1),\dots,c(y_m))$
can be computed in polynomial time from the values of the atomic
formulas. We conclude that computing the value of $\alpha(\bar v,\bar
\rho,\lambda,c(y_1),\dots,c(y_m))$ can be done in polynomial time.

There is, however, one small issue that requires explanation in order
for the previous procedure to work properly. Assume that the procedure
``guesses'' the interpretation of a variable $y_i$ in $\phi_\tau$ to
be the encoding of a path in $G$ from $v$ to $v'$ that satisfies $\E \subseteq \{e_1,\dots,e_m\} \times D_\bot^k \times D_\bot^k$. Then it is necessary to check that, if
the encoding implies that this path is $\rho^i_{\E,v,v'}$, then $i
\leq c_{\E,v,v'}$. In order to do so, the procedure needs to check, in
a subroutine, whether there exist $i$ different paths from $v$ to $v'$
that satisfy $\E$.  The next claim shows that this can be done in
exponential space, which finishes the proof of the theorem.

\medskip 

\begin{claim} For each pair $v,v' \in V$, $\E \subseteq
  \{e_1,\dots,e_m\} \times D_\bot^k \times D_\bot^k$, and $i \leq k + |\bar \rho|$, one can check in
  \expspace\ whether there are $i$ distinct paths in $G$
  from $v$ to $v'$ that satisfy $\E$.  \end{claim}

\begin{proof} Let $\#$ be a symbol not in $\Sigma$ and denote by 
  $\Sigma_\#$  the alphabet $\Sigma \cup \{\#\}$. 
Let $\A_{v,v'}$ be the automaton over alphabet $\{v\}
\cup (\Sigma_\# \times V)$ defined as follows.  The set of states is
the disjoint union of $V$ with a new state $s$. The initial state of
$\A$ is $s$ and the final state is $v'$. Further, the transition
relation of $\A$ is defined as follows: (1) For every edge
$(v_1,a,v_2) \in E$ there is a transition in $\A$ from $v_1$ to $v_2$
labeled $(a,v_2)$, (2) for every node $v_1 \in V$ there is a
transition from $v_1$ to $v_1$ in $\A$ labeled $(\#,v_1)$, and (3)
there is a transition in $\A$ from $s$ to $v$ labeled
$v$. Intuitively, $\A_{v,v'}$ accepts exactly those strings of the
form $v (a_1,v_1) (a_2,v_2) \cdots (a,v')$ such that $v a_1 v_1 a_2
v_2 \cdots a v'$ is a path in $G$ from $v$ to $v'$, when we allow
paths to loop arbitrarily many times on $\#$-labeled nodes.

Let $\A^i_{v,v'}$ be the automaton over alphabet $\{v^i\} \cup
(\Sigma_\# \times V)^i$ defined as follows: The set of states is
$V^i \cup \{s^i\}$, the initial state is $s^i$ and the final state is
$(v')^i$. There is a transition in $\A^i$ from $\bar u =
(u_1,\dots,u_i)$ to $\bar w = (w_1,\dots,w_i)$ labeled $\bar t =
(t_1,\dots,t_i)$ iff there is a transition labeled $t_\l$ from $u_\l$
to $w_\l$ in $\A_{v,v'}$, for each $1 \leq \l \leq i$. (Notice that 
$\A^i_{v,v'}$ is not exactly the $i$-th product of $\A_{v,v'}$ with itself, 
as $\A^i_{v,v'}$ does not contain all states in such a product). 
Clearly,
$\A^i_{v,v'}$ is of exponential size but the size of each one of its
states is polynomial. Furthermore, it is decidable in polynomial time
whether there exists a transition labeled $\bar t$ from state $\bar u$
to $\bar w$ in $\A^i_{v,v'}$.

Define now an automaton $\A'_{v,v'}$ that is the restriction of
$\A^i_{v,v'}$ to those strings 
\[
v^i (w^1_1,\dots,w^1_i) \cdots
(w^p_1,\dots,w^p_i)
\] over alphabet $\{v^i\} \cup (\Sigma_\# \times
V)^i$ that satisfy the following:

\begin{itemize}

\item For each $1 \leq \l \leq i$, if for some $1 \leq j < p$ it is
the case that $w^j_\l = (\#,v')$, for some $v' \in V$, then for
each $j < k \leq p$ it is the case that $w^k_\l = (\#,v')$.

\item For each $1 \leq \l,t \leq i$, if $\l \neq t$ then the strings
$v w^1_\l \cdots w^p_\l$ and $v w^1_t \cdots w^p_t$ over alphabet
$\{v\} \cup (\Sigma_\# \times V)$ are different.

\end{itemize}
The first condition says that each projection of a string accepted by
$\A'_{v,v'}$ represents a path in $G$ from $v$ to $v'$ that loops only
on $v'$ and only at the end of the path. The second condition ensures
that any two distinct projections of a path accepted by $\A'_{v,v'}$
represent different strings. 

It is not hard to prove that the language accepted by $\A'_{v,v'}$ is
nonempty iff there exist $i$ distinct paths in $G$ from $v$ to
$v'$. Further, it is also not hard to see that $\A'_{v,v'}$ is of
exponential size but the size of each one of its states is polynomial;
and it is decidable in polynomial time whether there exists a
transition labeled $\bar t$ from state $\bar q$ to state $\bar q'$ in
$\A'_{v,v'}$.

Using techniques in \cite{LV}, it is also possible to construct in
exponential time an NFA $\A_{v,v',e_i,\lambda,\lambda'}$ ($e_i \in
\{1,\dots,m\}$, $\lambda,\lambda' \in D_\bot^k$)
over alphabet $v \cup (\Sigma_\# \times V)$, 
that accepts precisely the strings $w$ accepted by $\A_{v,v'}$ such
that the path
$\rho$ from $v$ to $v'$ in $G$ 
that is represented by $w$ satisfies  
$e_i(\rho,\lambda,\lambda')$. (The main idea 
is to construct $\A_{v,v',e_i,\lambda,\lambda'}$ in such a way 
that, at each position while reading $w$, it keeps in its state 
the $k$-tuple of data values that is stored in the registers of $e_i$). 
The set of states of 
$\A_{v,v',e_i,\lambda,\lambda'}$ is of exponential size, but each
particular state can be represented using polynomial space. Further, deciding
if there is a transition between two states of   
$\A_{v,v',e_i,\lambda,\lambda'}$ can be done in polynomial time. 
This means that for each $\A_{v,v',e_i,\lambda,\lambda'}$, its
complement can be constructed in
double exponential time, with each state using only exponential space. 

It is not hard to see, then, that one can construct 
in double exponential time 
an automaton $\A_{\E,v,v'}$ over alphabet $\{v^i\} \cup
(\Sigma_\# \times V)^i$ 
that does the following: It starts from
$\A'_{v,v'}$, and restricts acceptance to strings 
 $v^i (w^1_1,\dots,w^1_i) \cdots$ 
$(w^p_1,\dots,w^p_i)$
over alphabet $\{v^i\} \cup
(\Sigma_\# \times V)^i$ such that:
\begin{itemize}
\item For each $1 \leq j \leq i$ and tuple $(e_\l,\lambda,\lambda')
\in E$, it is the case that  
$v w_j^1 \cdots w_j^p$ is accepted by
$\A_{v,v',e_\l,\lambda,\lambda'}$.   
\item 
For each $1 \leq j \leq i$ and tuple $(e_\l,\lambda,\lambda')
\not\in E$, it is the case that  
$v w_j^1 \cdots w_j^p$ is accepted by
the complement $\A_{v,v',e_\l,\lambda,\lambda'}$. 
\end{itemize}
Further, each state in $\A_{\E,v,v'}$ can be represented using
exponential space and checking whether there is a transition between
two given states of $\A_{\E,v,v'}$ can be done in polynomial
time. 

It is clear that there exist $i$ distinct paths in $G$
  from $v$ to $v'$ that satisfy $\E$ if and only if $\A_{\E,v,v'}$
accepts at least one string. But we can check $\A_{\E,v,v'}$ for
nonemptiness in \expspace\ using a standard ``on-the-fly" argument. 
This finishes the proof of the claim.
\end{proof} 

This also finishes the proof that {\sc Eval}{\em (RL)} is in {\sc
  Expspace}. Now we show that {\sc Eval}{\em (RL)} is {\sc
Expspace}-hard. 

For all constants $f_0$, we provide a reduction from the class of problems solvable by a Turing
machine using a tape of size $2^{f_0 n}$ given an input word of size
$n$. There are a Turing machine $M$ and a constant $f_0$ such that the following problem is {\sc
  ExpSpace}-hard: given a word $w$ of size $n$, is there an accepting
run of $M$ over $w$ using at most $2^{f_0 n}$ cells? We prove that there is a formula $\phi
\in$ RL such that for all words $w$ of size $n$, there are a
formula $\phi_w$ and a graph
$G_w$ such that
\[
G_w \vDash \phi_w \quad \text{iff} \quad \text{there is an accepting
run of $M$ over $w$ using at most $2^{f_0 n}$ cells}. 
\] 
Let $(\Sigma,Q, \delta,q_0,q_f)$ be the Turing machine $M$,
where $\Sigma$ is the alphabet consisting of the input alphabet and
the blank symbol $B$, $q_0$ is the initial state, $q_f$ is the final state and
$\delta: Q \times \Sigma \to (Q \times \Sigma \times \{ L,R\})$ is the
transition map, where $L$ stands for ``left'' and $R$ stands for ``right''.

The formula $\phi_w$ that we associate with the machine $M$ and a 
word $w$ is a formula
of the form
\[
\exists \pi \psi_w(\pi),
\]
where $\psi_w$ is a formula that does not contain any quantification
over path variables. The formula $\psi_w(\pi)$ expresses that the path $\pi$ in the graph
$G_w$ encodes an accepting run of $M$ over the word $w$. 

As in the proof of Theorem~\ref{theo:main}, we encode a configuration
$C$ of a run of $M$ in the following way. Suppose that the content of
the tape is the word $w'=w'_1 \dots w'_{2^{f_0 n}}$, the
head is scanning the cell number $i_0$ and the machine is
in state $q$.  We encode 
the configuration $C$ by the word $e_C= d^C_1 \& \dots d^C_{2^{f_0 n}} \&$
where $\&$ plays the role of a delimiter and  each $d^C_i$
 encodes the information in cell number $i$. More precisely, if
 $\Delta$ is the alphabet $ (Q \cup
\{ \$ \}) \times \Sigma$, we define $d_i^C$
as  the word
\[
c(i)  \; (q'_i,w'_i),
\]
where $c(i)$ and $q'_i$ are defined as follows. The word $c(i)$
is the binary representation of the number $i$. The
 letter $w'_i$ is the content of the cell $i$. The letter $q'_i$ is equal
 to $\$$ if the head is not scanning the cell number $i$; otherwise,
 $q'_i$ is equal to the state $q$. That is, $q'_{i_0} = q$
 and for all $i \neq i_0$, $q'_i = \$$. 
We encode a run $C_0 C_1 \dots$  as  the word $e_{C_0} \#
e_{C_1} \# \dots$. We define the formula $\psi_w(\pi)$ and the
graph $G_w$ in such a
way that a path $\pi$ satisfies $\psi_w$ iff the label of 
 $\pi$ is the
encoding of an accepting run of $M$ over $w$.    

The graph $G_w$ is (almost) the same graph as the graph in the proof of
Theorem~\ref{theo:main} in the case where $k=1$. That is, $G_w$ is
obtained by ``linking'' two graphs $I_w$ and $H$. If $C_0 C_1 \dots$
is an accepting run with associated path $\pi$, the trace of $\pi$
on $I_w$ will correspond to the encoding of the initial configuration
$C_0$, while the trace of $\pi$ on $H$ is the encoding of the run $C_1
C_2 \dots$. 
The graph $H$ is
given by
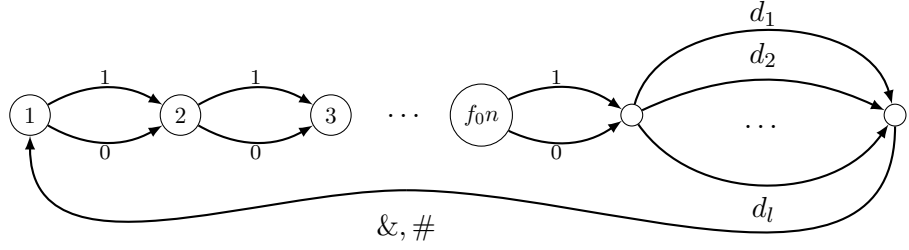
\begin{center}
\begin{tikzpicture} [scale=0.5]
\node(1) [draw,circle,scale=0.8] at (0,0) {$1$};
\node(Y) at (2,1) {{\scriptsize $1$}};
\node(N)at (2,-1) {{\scriptsize $0$}};
\node(2) [draw, circle,scale=0.8] at (4,0) {$2$};
\draw [->,>=latex, thick]
    (1) to[bend left] (2);
\draw [->,>=latex, thick]
    (1) to[bend right] (2);

 \node(Y2) at (6,1) {{\scriptsize $1$}};
\node(N2) at (6,-1) {{\scriptsize $0$}};

\node(3) [draw,circle,scale=0.8] at (8,0) {$3$};
\draw [->,>=latex, thick]
    (2) to[bend left] (3);
\draw [->,>=latex, thick]
    (2) to[bend right] (3);

\node(d) at (10,0) {$\dots$};

\node(n)  [draw, circle,scale=0.8] at (12,0) {$f_0 n$};
 \node(Yn) at (14,1) {{\scriptsize $1$}};
\node(Nn) at (14,-1) {{\scriptsize $0$}};
\node(f)  [draw, circle,scale=0.8] at (16,0) {};
\draw [->,>=latex, thick]
    (n) to[bend left] (f);
\draw [->,>=latex, thick]
    (n) to[bend right] (f);

\node(f2) [draw,circle,scale=0.8] at (23,0) {};
\draw [->,>=latex, thick]
    (f) to[bend left =75] (f2);
\draw [->,>=latex, thick]
    (f) to[bend left=25] (f2);
\draw [->,>=latex, thick]
    (f) to[bend right=55] (f2);
\node(d1) at (19.5,2.7) {$d_1$};
\node(d2) at (19.5,1.6) {$d_2$};
\node(d3) at (19.5,-2.5) {$d_l$};
\node at (19.5,-0.3) {$\dots$};
\node at (10,-3) {$\&, \#$};
\node(r)[anchor=east] at (10,-2){};
\node(r2)[anchor=west] at (10,-2){};
\draw [thick]
    (f2) to[out=270,in=0] (r);
\draw [->,>=latex, thick] (r2) to[out=180,in=270] (1);
\end{tikzpicture}
\end{center}
Recall that the set $\{ d_i : 1\leq i \leq l\} $ is defined as $(Q \cup \{\$\}) \times
\Sigma$. 
Consider a simple path $\pi'$ starting from the node with data value $1$ and ending in the node
with outgoing edges with labels $\&$ and $\#$.  Its label is of the
form $c(i) (q',a)$; that is, it is the encoding of the information in
a cell with number $i$. Hence, the label of a path $\pi$ starting and ending
in the node with data value $1$ satisfies $c(i_1) d_1 (\& \vee \#) c(i_2)
d_2
(\& \vee \#) \dots c(i_k) d_k (\& \vee \#)$, where each $d_j$ is the encoding
of the information of a cell and each $c(i_j)$ is the encoding of the
number $i_j$. We define the formula $\psi$ in such a
way  that if $\psi(\pi)$ holds, the succession of encodings of cells
describe a run of $M$.

Next we define the graph $I_w$ where we encode the initial configuration. Suppose that $w=w_1 \dots w_n$. For
all $1 \leq i \leq n$, we introduce a graph $K_i$ describing the cell
number $i$ in the initial configuration. If  $b_1 \dots b_n$ is the
binary encoding of the number $i$, the graph $K_i$ is given by
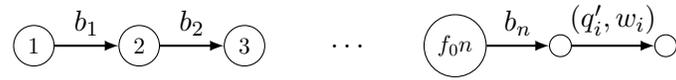
\begin{center}
\begin{tikzpicture} [scale=0.7]
\node(1) [draw,circle,scale=0.8] at (0,0) {$1$};

\node(2) [draw, circle,scale=0.8] at (2,0) {$2$};
\draw [->,>=latex, thick]
    (1) -- (2)
node[midway,above]{$b_1$}; 

\node(3) [draw,circle,scale=0.8] at (4,0) {$3$};
\draw [->,>=latex, thick]
    (2) -- (3) node[midway,above]{$b_2$}; 

\node(d) at (6,0) {$\dots$};

\node(n)  [draw, circle,scale=0.8] at (8,0) {$f_0 n$};
 \node(Yn) [draw, circle,scale=0.8] at (10,0) {};
\draw [->,>=latex, thick]
    (n) -- (Yn) node[midway,above]{ $b_n$}; 
\node(f)  [draw, circle,scale=0.8] at (12,0) {};
\draw [->,>=latex, thick]
    (Yn) -- (f) node[midway,above]{ $(q'_i,w_i)$}; 
\end{tikzpicture}
\end{center} 
where $q'_1 =q_0$ and $q'_i = \$$ if $i \neq 1$. The label of the
longest path of $K_i$ is exactly the encoding of cell number $i$ in
the initial configuration.  Next we define the graph $K$ which allows
us to encode the cells with number $>n$ in the initial
configuration. The graph $K$ is given by
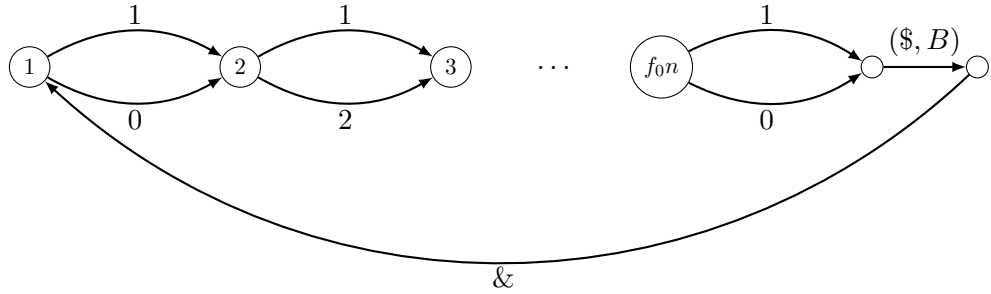
\begin{center}
\begin{tikzpicture} [scale=0.7]
\node(1) [draw,circle,scale=0.8] at (0,0) {$1$};
\node(Y) at (2,1) {$1$};
\node(N)at (2,-1) {$0$};
\node(2) [draw, circle,scale=0.8] at (4,0) {$2$};
\draw [->,>=latex, thick]
    (1) to[bend left] (2);
\draw [->,>=latex, thick]
    (1) to[bend right] (2);

 \node(Y2) at (6,1) {$1$};
\node(N2) at (6,-1) {$2$};

\node(3) [draw,circle,scale=0.8] at (8,0) {$3$};
\draw [->,>=latex, thick]
    (2) to[bend left] (3);
\draw [->,>=latex, thick]
    (2) to[bend right] (3);

\node(d) at (10,0) {$\dots$};

\node(n)  [draw, circle,scale=0.8] at (12,0) {$f_0 n$};
 \node(Yn) at (14,1) {$1$};
\node(Nn) at (14,-1) {$0$};
\node(f)  [draw, circle,scale=0.8] at (16,0) {};
\draw [->,>=latex, thick]
    (n) to[bend left] (f);
\draw [->,>=latex, thick]
    (n) to[bend right] (f);
\node(d) [draw, circle,scale=0.8] at (18,0) {};
\draw [->,>=latex, thick]
    (f) -- (d) node[midway,above]{ $(\$,B)$}; 
\draw [->,>=latex, thick]
    (d) to[bend left=43] (1) ;
\node at (9,-4) { $\&$};
\end{tikzpicture}
\end{center}
The label of a simple path from the node with data value $1$ to the node with
outgoing edge with label $\&$ is of the form $c(i) (\$,B)$; that is,
it is the encoding of an unscanned cell with a blank symbol. In
particular, it is the encoding of all the cells with number $>n$ in
the initial configuration. 
The graph $I_w$ is obtained by linking together the graph 
$K_1,\dots,K_n,K$ in the following way 

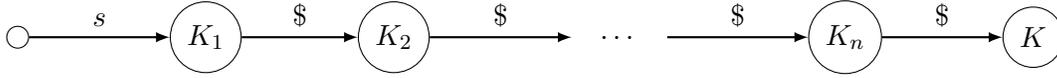
\begin{center}
\begin{tikzpicture} [scale=0.5]
\node(s) [draw, circle,scale=0.8] at (-5,0) {};
\node(0) [draw, circle,scale=1] at (0,0) {$K_1$};
\node(1) [draw, circle,scale=1] at (5,0) {$K_2$};
\draw [->,>=latex, thick]
    (0) -- (1) node[midway,above]{ $\$$};
\draw [->,>=latex, thick]
    (s) -- (0) node[midway,above]{ $s$};

\node(2) at (10,0) {};
\draw [->,>=latex, thick]
    (1) -- (2) node[midway,above]{ $\$$};

\node(d) at (11,0) {$\dots$};
\node(n1) at (12,0) {};
\node(n) [draw, circle,scale=1] at (17,0) {$K_{n}$}; 

\draw [->,>=latex, thick]
    (n1) -- (n) node[midway,above]{ $\$$};
\node(i) [draw, circle,scale=1] at (22,0) {$K$}; 
\draw [->,>=latex, thick]
    (n) -- (i) node[midway,above]{ $\$$};
\end{tikzpicture}
\end{center}

\noindent The arrow with a label $s$ is an arrow pointing to the node
of $K_1$ with data value $1$. Each arrow with label $\#$ between two graphs is from the
``right-most'' node of the first graph to the ``left-most'' node of
the second graph. 
Finally, the graph $G_w$ is defined by 
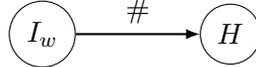
\begin{center}
\begin{tikzpicture} [scale=0.5]
\node(0) [draw, circle,scale=1] at (0,0) {$I_w$};
\node(1) [draw, circle,scale=1] at (5,0) {$H$};
\draw [->,>=latex, thick]
    (0) -- (1) node[midway,above]{ $\#$};
\end{tikzpicture}
\end{center}
where the edge with label $\#$ is an edge from the ``right-most'' node
of $K$ to the node with data value $1$ in $H$. 
We define now a formula $\psi_\pi$ such that
\begin{multline*} 
G_w \vDash \psi(\pi) \quad \text{iff}\\ 
 \ l(\pi) \text{ is the
  encoding of an accepting run of $M$ over $w$ using at most $2^{f_0 n}$ cells,}
\end{multline*}
where $l(\pi)$ is the label of $\pi$. 
In fact it is easier to define a formula $\chi(\pi)$ such that
\begin{multline} \label{eq:chiexp}
G_w \vDash \chi(\pi) \quad \text{iff} \\ l(\pi) \text{ is not the
  encoding of an accepting run of $M$ over $w$ using at most $2^{f_0 n}$ cells.}
\end{multline}
Suppose that $\pi$ is a path through $G_w$. The label of $\pi$ is
not the encoding of an accepting run of $M$ over $w$ using at most
$2^{f_0 n}$ cells iff at least one of the
following conditions holds.

\begin{enumerate}[label=(\roman*)]
\item The first letter of $l(\pi)$ is not the initial letter $s$ or the run never reaches a final state,~i.e. there is no pair
  of the form $(q_f,a)$ for some $a$, occurring in $l(\pi)$.
\item The symbol $\#$ is not at the ``right place'':
\begin{itemize}[label=$-$]
\item either after we reach the symbol $\#$ (i.e. we are going to
  enter the encoding of a new configuration), the label contains
  the binary encoding of a number $\neq 1$,
\item or after the binary encoding of the number $2^{f_0 n}$ (that is,
  after encoding the information of the last cell), the
  symbol $\#$ does not occur in the label of $\pi$. Since $\#$ is used as a delimiter between encodings of
  configurations, this means that although we finished encoding 
   the last cell of a configuration, we do not move to a new configuration. 
\end{itemize}
\item There is a substring $c(i) d \& c(j) d'$ (where $i<2^n$) such that $j$ is
  not the successor of $i$. That is, after encoding the information of
  cell number $i$, we do not encode the information of cell number
  $i+1$.  
\item Finally, there is a string  $e_C \# e_{C'}$ of
  $D_\pi$ such that $C$ and $C'$ are not successive configurations.
\end{enumerate}

\noindent Expressing cases~(i) and~(ii) is fairly easy. We concentrate on~(iii)
and~(iv). We start by showing how to express case~(iv). Suppose that $c(i)$ and $c(j)$ are  two successive
  binary encodings occurring in the label of $\pi$. Suppose
  $c(i)=b_1 \dots b_{f_0 n}$ and $c(j) = b'_1 \dots b'_{f_0 n}$. Then $j$ is
  not the successor if $i$ iff one of the following holds.
\begin{enumerate}[label=\({\alph*}]
\item For some $k$, $b_{k} \dots b_{f_0 n}$ is equal to $1 \dots 1$ and
  $b'_k$ is not equal to $0$. 
\item For  some $k$, $b_{k} \dots b_{f_0 n}$ is equal to $0 1 \dots 1$ and
  $b'_k$ is not equal to $1$. 
\item For some  $k$, $b_{k} = 1$, $0$ occurs in  $b_{k+1}\dots
  b_{f_0 n}$ and
  $b'_k$ is not equal to $1$. 
\item For some  $k$, $b_{k} = 0$, $0$ occurs in  $b_{k+1}\dots
  b_{f_0 n}$ and
  $b'_k$ is not equal to $0$. 
\end{enumerate}
We show how we can express~(a); the other cases are similar. Case~(a)
is expressed by the following REM
\[
(\Lambda\backslash \Delta_f)^*  \downarrow r_1. 1^*  \Delta \& \{ 0,1\}^*  [r_1^=] 1  \Lambda^*.
\]
where $\Delta_f$ is the set $\{ q_f\} \times \Sigma $ and $\Lambda$ is
the alphabet of the graph $G_w$. 
In the register $r_1$, we store the number $k$ such that $b_{k} \dots
b_{f_0 n} = 1 \dots 1$ (that is, we only go through edges with label $1$ until we reach an edge
with label in $\Delta$). When we reach again the node with data value
$k$, the label of the outgoing edge is $1$, expressing that
$b'_{k}=1$.  

Finally we look at case~(iv), i.e. how to express that there is a string  $e_C \# e_{C'}$ of
  $D_\pi$ such that $C$ and $C'$ are not successive configurations. This might happen for several
  reasons:~(A) either we did not modify properly the content of a cell
  or move properly the head or~(B) we modified the content that was
supposed to remain constant. We only treat 
case~(A), as the other case can be handled in a similar way (note that
the proof of Theorem~\ref{theo:rl-dc} is very similar to this proof
and there, we will treat case~(B)). In case~(A), we also only consider the case of a transition
moving the head to the right, the other case being symmetric.  

So suppose that $\delta(q,a)=(q',b,R)$. As in the proof of
Theorem~\ref{theo:main}, we use a slightly different definition of a
run of a Turing machine (which is equivalent to the usual definition, but helps us to keep our
formulas simpler). We assume that if $\delta(q,a)=(q',b,R)$ and the machine scans a cell
$c$ with content $a$, then in the next
state, the machine scans the successor $c'$ of $c$, the content of
$c'$ is $b$, while the content of $c$ is $a$. 

Let $(\Lambda_{!\#})^*$
  be the set of words over $\Lambda$ that contain at most one
  occurrence of $\#$.  We define $e^R_{(q,a)}$
    as the following REM  
\[
\Lambda^* \{ 0,1 \} \downarrow r_1.  \dots \{ 0,1\} \downarrow r_n. 
(q,a) (\Lambda_{!\#})^* \{ 0,1\} [r_1^=] \dots  \{ 0,1\} [r_{f_0 n}^=] 
\Delta \& \{ 0,1\}^* (\Delta \backslash \{ (q',b) \}) \Lambda^*.
\]
We store in the registers $r_1, \dots, r_{f_0 n}$ the binary encoding of a
number $i$. That number is the number of a scanned cell with content $a$  and the
 current state is $q$.  After the next occurrence of $\#$ (after
 reading a word in $(\Lambda_{!\#})^*$), we enter a new
 configuration. In that new configuration, we reach the
 cell number $i$ when we read a sequence matching the contents of the
 registers. The encoding of the next cell (after reading a sequence
 in $\Delta \&$) must consist of the binary encoding of a number
 followed by a symbol in $\Delta$ that is not $(q',b)$.

If $\Delta_R$ is the set $\{ (q,a) : \delta(q,a)=(q',b,R)  \text{ for
  some } q', b \}$ and $\lambda$ is a register of size $f_0 n$, we define the formula
\[
\bigvee \{ \exists \nu \; e^R_{(q,a)} (\pi,\bot,\nu) : (q,a) \in \Delta_R \}, 
\]
taking care of case~(iv)(A).
\end{proof}

\medskip 

The increase in expressiveness of RL over REM has an important cost in data
complexity, which becomes intractable: 

\begin{theorem} \label{theo:rl-dc} 
{\sc Eval}{\em (RL)} is in {\sc Pspace} in data
complexity. Furthermore, 
there is a finite alphabet $\Sigma$ and a RL formula $\phi$ over $\Sigma$ 
and a single register $r$,  
such that {\sc Eval}{\em (RL,$\phi$)} is {\sc Pspace}-hard. In
addition, the latter holds even if the input is restricted to graph
databases.  
\end{theorem} 

\begin{proof}
The upper bound follows as a corollary to the proof of the upper bound
in Theorem \ref{theo:rl-cc}. In fact, it is clear that the whole
process can be carried in \pspace\ if we assume a fixed RL query (in
fact, to obtain a \pspace\ upper bound we do not need more than to fix
the number of registers used in the query).

For the lower bound, we define a formula $\phi$ in $\rl$ such that for
all constants 
$f_0 $, there is a reduction from the  class of problems solvable by a Turing
machine using a tape of size $f_0 n$ given an input word of size
$n$, to the evaluation problem of $\phi$. More precisely, there are a
Turing machine $M$ and a constant $f_0$ such that the following problem is {\sc PSpace)}-hard: given a word $w$ of size $n$, is there an accepting
run of $M$ over $w$ using at most $f_0 n$ cells? 

We prove that the formula $\phi$ is such that for all words $w$ of size $n$, there is a graph
$G_w$ such that
\[
G_w \vDash \phi \quad \text{iff} \quad \text{there is an accepting
run of $M$ over $w$ using at most $f_0 n$ cells}. 
\] 
Let $(\Sigma,Q, \delta,q_i,q_f)$ be the Turing machine $M$,
where $\Sigma$ is the input alphabet together with a blank symbol $B$, $q_0$ is the initial state, $q_f$ is the final state and
$\delta: Q \times \Sigma \to (Q \times \Sigma \times \{ L,R\})$ is the
transition map, where $L$ stands for ``left'' and $R$ stands for ``right''.

The formula $\phi$ that we associate with the machine $M$ is a formula
of the form
\[
\exists \pi \psi(\pi),
\]
where $\psi$ is a formula that does not contain any quantification
over path variables. Given a word $w$, the path $\pi$ in the graph
$G_w$ will encode an accepting run of $M$ over the word $w$. 

Given a word $w$ of size $n$, consider a configuration $C$ of the run of
$M$ over $w$ where the contents of the tape is the word $w'=w'_1 \dots w'_{f_0 n}$, the
head is scanning the cell number $i_0$ and the machine is
in state $q$.  Similarly to the proof of Theorem~\ref{theo:main},  we encode 
the configuration $C$ by the word 
\[
e_C= N d^C_1 N  
d^c_2 \dots N  d^C_{f_0 n}
\]
 where each $d^C_i$
 encodes the information in cell number $i$ in the configuration $C$. 
We define $d_i^C$
as the pair $(q'_i,w'_i)$, 
where $q'_i$ is defined as follows. The
 letter $w'_i$ is the contents of the cell $i$. The letter $q'_i$ is equal
 to $\$$ if the head is not scanning the cell number $i$; otherwise,
 $q'_i$ is equal to the state $q$. That is, $q'_{i_0} = q$
 and for all $i \neq i_0$, $q'_i = \$$. 
  
The run of $M$ over the word $w$ is a (possibly infinite) sequence of configurations of
the form $C_0 C_1 \dots$. We encode the run as  the word $e_{C_0} \#
e_{C_1} \# \dots$, where $\#$ plays the role of a delimiter. We will define the formula $\psi(\pi)$ and the
graph $G_w$ in such a 
way that a path $\pi$ satisfies $\psi$ iff the label of 
$\pi$ is the
encoding (as defined above) of an accepting run of $M$ over $w$.    

We think of a path $\pi$ encoding a run of $M$ over $w$ as consisting
of two parts. The label of the first part contains the encoding $ e_{C_0}$ of the
initial configuration $C_0$. The label of the second part contains the encoding  $e_{C_1} \#
e_{C_2} \# \dots$ of the remaining part of the run. The first part of the
path $\pi$ is a path in a subgraph $I_w$ of  $G_w$, while the second
part is a path in the subgraph $H$ (independent of $w$) of $G_w$. The graph $G_w$ will be
obtained by adding an  edge from a node in $I_w$ to a node in $H$. 

The graph
$I_{w}$ is given by
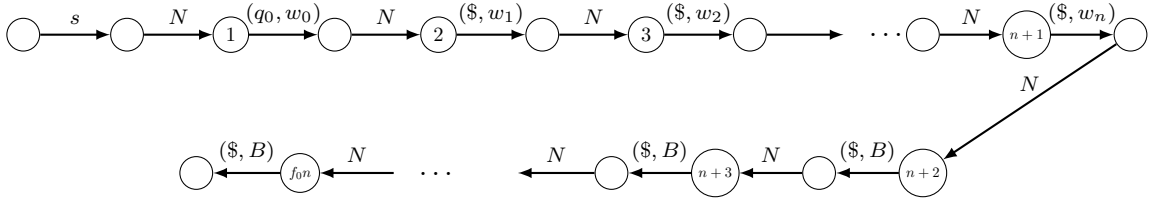
\begin{center}
\begin{tikzpicture} [scale=0.46]
\node(i) [draw, circle,scale=1.2] at (-3,0) {};
\node(S) [draw, circle,scale=1.2] at (0,0) {};
\node(1) [draw, circle,scale=0.7] at (3,0) {$1$};
\node(V1) [draw, circle,scale=1.2] at (6,0) {};
\node(2) [draw, circle,scale=0.7] at (9,0) {$2$};
\node(V2) [draw, circle,scale=1.2] at (12,0) {};
\node(3) [draw, circle,scale=0.7] at (15,0) {$3$};
\node(V3) [draw, circle,scale=1.2] at (18,0) {};
\node(B) at (21,0) {};
\node at (22,0) {$\dots$};
\node(C) [draw, circle,scale=1.2]  at (23,0) {};
\node(n) [draw, circle,scale=0.5]  at (26,0) {$n+1$};
\node(Vn) [draw, circle,scale=1.2] at (29,0) {};

\node(n1) [draw, circle,scale=0.5]  at (23,-4) {$n+2$};
\node(n11) [draw, circle,scale=1.2] at (20,-4) {};
\node(n2) [draw, circle,scale=0.5]  at (17,-4) {$n+3$};
\node(n22) [draw, circle,scale=1.2] at (14,-4) {};
\node(B2) at (11,-4) {};
\node at (9,-4) {$\dots$};
\node(B3) at (8,-4) {};
\node(f) [draw, circle,scale=0.5]  at (5,-4) {$f_0 n$};
\node (f2) [draw, circle,scale=1.2]  at (2,-4) {};

\draw [->,>=latex, thick]
    (i) -- (S)
node[midway,above]{{\scriptsize $s$}}; 
\draw [->,>=latex, thick]
    (B3) -- (f)
node[midway,above]{{\scriptsize $N$}}; 

\draw [->,>=latex, thick]
    (f) -- (f2)
node[midway,above]{{\scriptsize $(\$,B)$}}; 

\draw [->,>=latex, thick]
    (Vn) -- (n1)
node[midway,above]{{\scriptsize $N$}}; 

\draw [->,>=latex, thick]
    (n1) -- (n11)
node[midway,above]{{\scriptsize $(\$,B)$}}; 

\draw [->,>=latex, thick]
    (n11) -- (n2) node[midway,above]{{\scriptsize $N$}}; 

\draw [->,>=latex, thick]
    (n2) -- (n22)
node[midway,above]{{\scriptsize $(\$,B)$}}; 

\draw [->,>=latex, thick]
    (n22) -- (B2)
node[midway,above]{{\scriptsize $N$}}; 

\draw [->,>=latex, thick]
    (C) -- (n)
node[midway,above]{{\scriptsize $N$}}; 

\draw [->,>=latex, thick]
    (n) -- (Vn)
node[midway,above]{{\scriptsize $(\$,w_n)$}}; 

\draw [->,>=latex, thick]
    (S) -- (1)
node[midway,above]{{\scriptsize $N$}}; 
\draw [->,>=latex, thick]
    (1) -- (V1)
node[midway,above]{{\scriptsize $(q_0,w_0)$}}; 

\draw [->,>=latex, thick]
    (V1) -- (2)
node[midway,above]{{\scriptsize $N$}}; 
\draw [->,>=latex, thick]
    (2) -- (V2)
node[midway,above]{{\scriptsize $(\$,w_1)$}}; 

\draw [->,>=latex, thick]
    (V2) -- (3)
node[midway,above]{{\scriptsize $N$}}; 
\draw [->,>=latex, thick]
    (3) -- (V3)
node[midway,above]{{\scriptsize $(\$,w_2)$}};
\draw [->,>=latex, thick]
    (V3) -- (B);
\end{tikzpicture}
\end{center}

In the graph above the data value $i$ carried by a node $v$ indicates
that the label of the outgoing edge of $v$ is $d_i^{C_0}$ (where
$C_0$ is the initial configuration and $d_i^{C_0}$ is defined as
above). Recall that $d_i^{C_0}$ indicates the contents of the cell 
number $i$ and whether or not the head is scanning that cell. There is
a unique path $\pi_0$ from the node with no incoming edge to the unique node with no
outgoing edge. The label of $\pi_0$ is
\[
s  N (q_i,w_0) \;  (\$,w_1) N(\$,w_2) \dots N
(\$,w_n) N(\$,B) \dots N(\$,B),
\]
that is, the word $s.e_{C_0}$, where $e_{C_0}$ is the encoding of the initial configuration $C_0$. 

We define now the graph $H$ encoding the remaining part of the run of
the machine. 

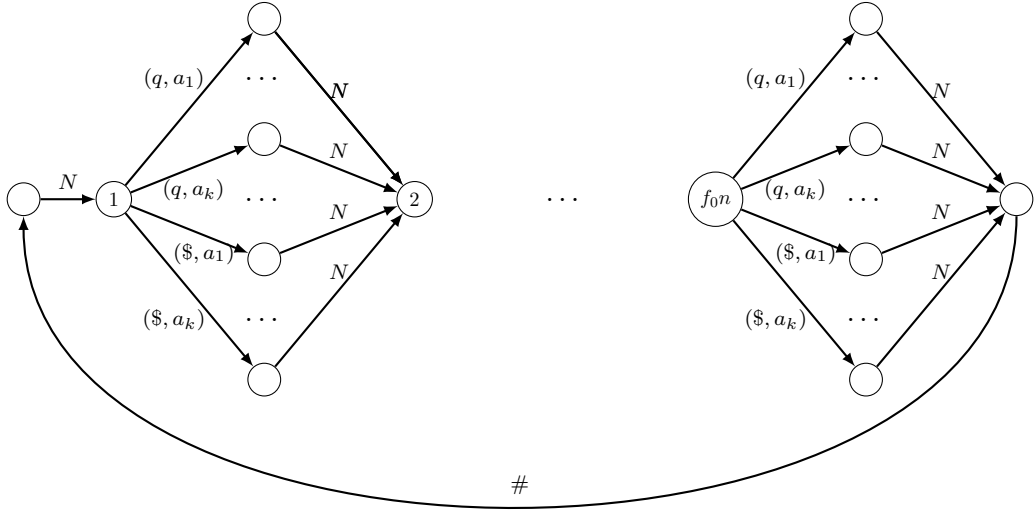
\begin{center}
\begin{tikzpicture} [scale=0.4]
\node(s)  [draw, circle,scale=1.2] at (0,0) {};
\node(1)  [draw, circle,scale=0.7] at (3,0) {$1$};
\node(f)  [draw, circle,scale=1.2] at (33,0) {};

\node(11)  [draw, circle,scale=1.2] at (8,6) {};
\node(1k)  [draw, circle,scale=1.2] at (8,2) {};

\node at (8,4) {$\dots$};
\node at (8,0) {$\dots$};
\node at (8,-4) {$\dots$};

\node(s1)  [draw, circle,scale=1.2] at (8,-2) {};
\node(sk)  [draw, circle,scale=1.2] at (8,-6) {};

\node(2)  [draw, circle,scale=0.7] at (13,0) {$2$};
\node(nk)  [draw, circle,scale=0.7] at (23,0) {$f_0 n$};
\node at (18,0) {$\dots$};

\node(21)  [draw, circle,scale=1.2] at (28,6) {};
\node(2k)  [draw, circle,scale=1.2] at (28,2) {};

\node at (28,4) {$\dots$};
\node at (28,0) {$\dots$};
\node at (28,-4) {$\dots$};

\node(s2)  [draw, circle,scale=1.2] at (28,-2) {};
\node(s2k)  [draw, circle,scale=1.2] at (28,-6) {};
\node at (16.5,-9.5) {{\scriptsize $\#$}};


\draw [->,>=latex, thick]
    (s) -- (1)
node[midway,above]{{\scriptsize $N$}};

\draw [->,>=latex, thick]
    (11) -- (2)
node[midway,above]{{\scriptsize $N$}};
\draw [->,>=latex, thick]
    (1k) -- (2)
node[midway,above]{{\scriptsize $N$}};
\draw [->,>=latex, thick]
    (s1) -- (2)
node[midway,above]{{\scriptsize $N$}};
\draw [->,>=latex, thick]
    (sk) -- (2)
node[midway,above]{{\scriptsize $N$}};

\draw [->,>=latex, thick]
    (21) -- (f)
node[midway,above]{{\scriptsize $N$}};
\draw [->,>=latex, thick]
    (2k) -- (f)
node[midway,above]{{\scriptsize $N$}};
\draw [->,>=latex, thick]
    (s2) -- (f)
node[midway,above]{{\scriptsize $N$}};
\draw [->,>=latex, thick]
    (s2k) -- (f)
node[midway,above]{{\scriptsize $N$}};

\draw [->,>=latex, thick]
    (1) -- (11);
\node at (5,4) {{\scriptsize $(q,a_1)$}};

\draw [->,>=latex, thick]
    (1) -- (1k);
\node at (5.65,0.3) {{\scriptsize $(q,a_k)$}};

\draw [->,>=latex, thick]
    (1) -- (s1);
\node at (5,-4) {{\scriptsize $(\$,a_k)$}};

\draw [->,>=latex, thick]
    (1) -- (sk);
\node at (6,-1.8) {{\scriptsize $(\$,a_1)$}};

\draw [->,>=latex, thick]
    (nk) -- (21);
\node at (25,4) {{\scriptsize $(q,a_1)$}};

\draw [->,>=latex, thick]
    (nk) -- (2k);
\node at (25.65,0.3) {{\scriptsize $(q,a_k)$}};

\draw [->,>=latex, thick]
    (nk) -- (s2);
\node at (25,-4) {{\scriptsize $(\$,a_k)$}};

\draw [->,>=latex, thick]
    (nk) -- (s2k);
\node at (26,-1.8) {{\scriptsize $(\$,a_1)$}};

\draw [->,>=latex, thick]
    (11) -- (2)
node[midway,above]{{\scriptsize $N$}};

\draw [->,>=latex, thick]
    (f) edge[out=270,in=270] (s);
\end{tikzpicture}
\end{center}
For all $1 \leq i \leq f_0 n$, all $q \in Q$ and all $a \in \Sigma$, the
node with data value $i$ admits outgoing edges with
label $(q,a)$ and $(\$,a)$. 
A path from the ``left-most'' node to the ``right-most'' node that
does 
not go through the edge 
with label $\#$ has a label  of the form
\[
N   d_1 N  d_2 \dots  N d_{f_0 n}N,
\]
where each $d_i$ belongs to $(Q \cup \{ \$\}) \times \Sigma$, that is, an encoding of a
configuration of the machine. 

Hence, a path $\pi'$ from   the ``left-most'' node to the ``right-most'' node
of $H$ 
(possibly going through the edge with label $\#$) has a label of the form
\[
e_{C_1} \# e_{C_2} \# \dots e_{C_l},
\]
where each $e_{C_i}$ is the encoding of a configuration of the
machine. 

We are now ready to define the graph $G_w$.
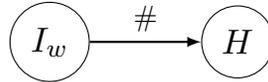
\begin{center}
\begin{tikzpicture} [scale=0.5]
\node(i)  [draw, circle,scale=1.2] at (0,0) {$I_w$};
\node(h)  [draw, circle,scale=1.2]  at (5,0) {$H$};
\draw [->,>=latex, thick]
    (i) -- (h)
node[midway,above]{$\#$};
\end{tikzpicture}
\end{center}
The edge with label $\#$ is an edge from the unique node in $I_w$ with
no outgoing edge to the ``left-most'' node in $H$. We define now the
formula $\psi$. Let $\Lambda$ be the alphabet $[(Q \cup \{ \$\}) \times
\Sigma] \cup \{ \#,s \}$. The formula $\psi$ must be such that 
\[
G_w \vDash \psi(\pi) \quad \text{iff} \quad l(\pi) \text{ is of the form } se_{C_0}\# e_{C_1} \dots e_{C_p} \Lambda^*,
\]
where $l(\pi)$ is the label of $\pi$ and $C_0 \dots
C_p$ is an accepting run of the machine over $w$. In fact, it will be
intuitively easier to first define a formula $\chi(\pi)$ such that
\begin{equation} \label{eq:psi}
G_w \vDash \chi(\pi) \quad \text{iff} \quad l(\pi) \text{ is not of the form } se_{C_0}\# e_{C_1} \dots e_{C_p} \Lambda^*,
\end{equation}
and define $\psi$ as $\neg \chi$. 
The formula
$\chi$ is obtained as a disjunction of the following subformulas.
\begin{itemize}
\item First $l(\pi)$ might not satisfy  $se_{C_0}\# e_{C_1} \dots e_{C_p} \Lambda^*$ because: (i)~it
  does not start with the letter $s$ or~(ii)~ it
  does not contain the encoding of a final configuration, i.e.~ it
  does not contain any occurrence of a pair of the form $(q_f,a)$ for
  some $a \in \Sigma$. Case~(i) is expressed by the formula 
\[
\bigvee \{ e_\gamma (\pi,\bot,\bot) : \gamma \in \Lambda, \gamma \neq s\},
\] 
where $e_\gamma$ is the REM 
$\gamma \Lambda^*$. Case~(ii) is expressed by the formula
\[
(\Lambda \backslash \Delta_f)^*) (\pi,\bot,\bot),
\]
where $\Delta_f$ is the set $\{ q_f\} \times \Sigma $. This formula 
express that there is no pair of the form $(q_f,a)$, for
  some $a \in \Sigma$.

\item Next $l(\pi)$ might not be of the ``right form'' because
  it contains a substring of the form $e_{C}\#e_{C'}$ occurring
  before a pair of the form $(q_f,a)$ and such that $C$ and $C'$ are
  not consecutive configurations. This might happen for several
  reasons:~(A) either we did not modify properly the contents of a cell
  or move properly the head or~(B) we modified the contents that was
supposed to remain constant.  

We will only treat one case. As we treated case~(A) in the proof of
Theorem~\ref{theo:rl-cc}, here we treat
case~(B). As in the proof  of
Theorem~\ref{theo:rl-cc}, we also consider a slightly different definition of a
run of a Turing machine (which is equivalent to the usual definition, but helps us to keep our
formulas simpler). We assume that if $\delta(q,a)=(q',b,R)$ and the machine scans a cell
$c$ with content $a$, then in the next
state, the machine scans the successor $c'$ of $c$, the content of
$c'$ is $b$, while the content of $c$ is $a$.

In case~(B), we make the following case distinction. 
\begin{itemize}
\item Suppose first that case~(B) happened because we modified the contents of the cell that was
  scanned by the head in configuration $C$. Note that where when moving from
  $C$ to $C'$, by definition of a run, we cannot modify the contents of
  the cell scanned in configuration $C$.

Let $(q,a)$ be a pair in $Q \times \Sigma$ and let
  $\Delta$ be the set $(Q \cup \{ \$\} \times \Sigma)$. We define also
  $\Lambda_{! \#}^*$ as the set of words over $\Lambda$ that contain
  at most 
  one occurrence of the symbol $\#$.
  We let $e_{(q,a)} (\pi)$ be the formula
\[
(\Lambda \backslash \Delta_f)^*  \downarrow r. (q,a) \Lambda_{! \#}^* [r^=]  (\Delta
\backslash \{(q,a):q 
\in Q \cup \{ \$ \} \}) 
\Lambda^*.
\]
Before we reach a final state, we store in register $r$ the number of a cell that is
scanned in the current configuration and with contents $a$. In the next
configuration (after reading a word in  $\Lambda_{! \#}^*$), when we
read the same cell, it does not contain
$a$. That is, the label of the edge is not a pair of the form $(q,a)$. 

The formula
\[
\bigvee \{ \exists \nu\; e_{(q,a)} (\pi,\bot,\nu) : (q,a) \in
Q \times \Sigma \},
\]
takes care of the cases where from moving from one configuration to
the next, we modified the contents of the cell scanned in the first
configuration. 

\item Next suppose that when moving from $C$ to $C'$, we modified the
  contents of a cell that was not scanned in $C$. Suppose also that
  according to the machine, we were not supposed to modified that
  contents. 

Let $a$ be a letter in $\Sigma$. Let $\Delta_L$ be the set of pairs $(q,a)$ such that
  $\delta(q,a)=(q',b,L)$ and let $\Delta_R$ be the set of pairs $(q,a)$ such that
  $\delta(q,a)=(q',b,R)$. We define $e_a(\pi)$ as the REM
\[
(\Lambda \backslash \Delta_f)^*  (\Delta \backslash \Delta_R)
\downarrow r . (\$,a) N (\Delta \backslash
\Delta_L)  \Lambda_{! \#}^* [r^=] (\Delta \backslash (\$,a)) \Lambda^*.
\] 
Before we reach a final state, in a configuration $C$, we store in register $r$ the number $i$
of an unscanned cell with contents
$a$. We assume that if the cell with number $(i-1)$ (if its exists) is
scanned in $C$, then the head is not moving to the right (i.e. to cell
number $i$) in the next configuration. That is, the pair describing
the cell number $(i-1)$ in $C$ is not a pair in $\Delta_R$. We also
assume that if cell number $(i+1)$ (if it exists) is scanned, then the
head is not moving to cell number $i$ in the next configuration,
i.e.~cell number $(i+1)$ is not described by a pair in $\Delta_L$.

Next we express that in the next configuration, it is not the case
that the cell with number $i$ is an unscanned cell with contents
$a$. That is, after reading  a word in  $\Lambda_{! \#}^* $, when we
see again the node stored in the register, the label of the edge is
not the pair $(\$,a)$. 

The formula
\[
\bigvee \{ \exists \nu \; e_a (\pi,\bot,\nu) : a \in \Delta\},
\]
takes care of case~(B) in the case where the cell modified is a cell
unscanned by the head. 
\end{itemize}

\end{itemize}
This finishes the proof of the theorem. 
\end{proof}

In the next section we introduce an interesting language, based on a
restriction of RL, that is tractable in data complexity, and thus
better suited for database applications. This language is a proper
extension of REM. 
But before that, we make some important remarks about the expressive power
of RL. 

\subsubsection*{Expressive power of RL}
We now look at the expressive power of the logic RL. It was proven in 
\cite{WL} that CRPQ is not subsumed by WL. Since RL subsumes CRPQ$^\neg$,
it follows that RL is not subsumed by WL. On the other hand, WL is also 
not subsumed by RL due to Theorem \ref{theo:main}, Theorem \ref{theo:rl-cc}, 
and the standard time/space hierarchy theorem from complexity theory. 
Therefore, we have the following proposition:
\begin{proposition}
The expressive power of WL and RL are incomparable.
\end{proposition}
On the other hand, we shall argue now that many natural queries about the
interaction between data and topology are also expressible in RL. The 
aforementioned query (Q) is one such example. We shall now mention other 
examples: hamiltonicity (H), the existence of an eulerian trail (E), 
bipartiteness (B), and connected graphs with an even number of nodes (C2).
The first two are expressible in WL, while (B) and (C2) are not known to
be expressible in WL. We conjecture that they are not. 

We now show how to express in RL the existence of a hamiltonian path in a graph;
the query (E) can be expressed in the same way but with two registers (to
remember edges, i.e., consisting of two nodes).
This is done with the following formula over $\Sigma = \{a\}$ and a single 
register $r$:
$$\exists \pi \: \big(\,\forall \lambda \forall \lambda' 
\neg e_1(\pi,\lambda,\lambda') \,
\wedge \, \forall \lambda (\lambda \neq \bot \rightarrow 
e_2(\pi,\lambda,\lambda))\,\big),$$ where $e_1 := a^* \cdot
(\downarrow \!\! r. a^+ [r^=]) \cdot a^*$ is the REM that checks
whether in a path some node is repeated (i.e., that it is not a simple
path), and $e_2 := a^*[r^=]a^*$ is the REM that checks that the node
stored in register $r$ appears in a path. 
In fact, this query expresses that there is a path $\pi$ that it is
simple (as expressed by the formula $\forall \lambda \forall \lambda'
\neg e_1(\pi,\lambda,\lambda')$), and
every node of the graph database is mentioned in $\pi$ (as expressed
by the formula $\forall \lambda (\lambda \neq \bot \rightarrow 
e_2(\pi,\lambda,\lambda))$). 

We now show how to express in RL the bipartiteness property from graph theory. 
An undirected graph $G = (V,E)$ is \emph{bipartite} if its set of 
nodes can be partitioned into two sets $S_1$ and $S_2$ such that, for each edge
$(v,w) \in E$, either (i) $v \in S_1$ and $w \in S_2$, or (ii) $v \in S_2$
and $w \in S_1$. It is well-known that a graph database is bipartite if and only if it 
does not
have cycles of odd length. The latter is expressible in RL since the existence 
of an odd-length cycle can be expressed
as $\exists \pi \exists \lambda \exists \lambda' e(\pi,\lambda,\lambda')$, where 
$e = \downarrow \! r. a(aa)^*[r^=]$. 

We now show how to express in RL that a graph database is a connected graph
with an odd number of nodes. To this end, it is sufficient and necessary
to express the existence of a hamiltonian path $\pi$ with an even number of
edges in the graph. But this is a simple modification
of our formula for expressing hamiltonicity: we add the check that
$\pi$ has an even number of edges by adding the conjunct $e(\pi,\nu,\nu')$,
where $e = (aa)^*$, and close the entire formula under existential
quantification of $\nu$ and $\nu'$.

\subsection{Tractability in data complexity} 

Let RL$^+$ be the {\em positive} fragment of RL, i.e., 
the logic obtained from
RL by forbidding negation but adding conjunctions (as they were not explicitly present 
in RL). 
It is easy to prove that the data complexity of the evaluation
problem for RL$^+$ is tractable (\nlogspace). This fragment
contains the class of {\em conjunctive} REMs, that has been previously
identified as tractable in data complexity \cite{LV}. However, 
the expressive power of RL$^+$ is limited as the
following proposition shows.

\begin{proposition} \label{prop:rlplus}
The query {\em (Q)} from the introduction is not expressible in RL$^+$. 
\end{proposition} 

\begin{proof}
Recall that $Q$ is the following query: Find pairs of
nodes $x$ and $y$ such that there is a node $z$ and a path $\pi$ from
$x$ to $y$ in which each node is connected to $z$. Suppose for contradiction that there is a formula $\phi$ in RL$^+$
over an alphabet $\Sigma$ and registers $r_1,\dots, r_k$, 
expressing $\exists x \exists y \, Q$. We may assume that $\phi $ is of the form
\[
\exists x_1 \dots \exists x_{n_1} \exists \pi_1 \dots \exists
\pi_{n_2} \exists \nu_1\dots \exists \nu_{n_3} \psi
\]
where $\psi$ is a disjunction of conjunctions of atoms. Let $G=(V,E,\kappa)$ be the
following graph
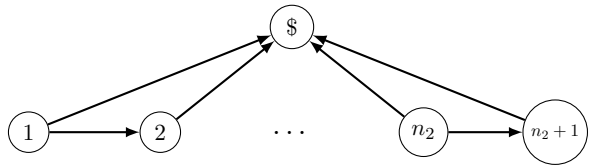
\begin{center}\begin{tikzpicture}[scale=0.7]
\node(s)  [draw, circle,scale=0.8] at (0,0) {$\$$};
\node(1)  [draw, circle,scale=0.8] at (-5,-2) {$1$};
\node(2)  [draw, circle,scale=0.8] at (-2.5,-2) {$2$};
\node(n)  [draw, circle,scale=0.8] at (2.5,-2) {$n_2$};
\node(n2)  [draw, circle,scale=0.6] at (5,-2) {$n_2+1$};
\node at (0,-2) {$\dots$};

\draw [->,>=latex, thick]
    (1) -- (s)
;
\draw [->,>=latex, thick]
    (2) -- (s)
;
\draw [->,>=latex, thick]
    (n2) -- (s)
;
\draw [->,>=latex, thick]
    (n) -- (s)
;
\draw [->,>=latex, thick]
    (1) -- (2)
;
\draw [->,>=latex, thick]
    (n) -- (n2)
;

\end{tikzpicture}
\end{center} 
where each edge is labeled with $a$.  The query
$\exists x \exists y \, Q$ is true in $G$; hence, the formula $\phi$ must be true in
$G$. That is, there is an assignment $\alpha$ mapping each variable
$x_i$ to a node in $G$, each path variable $\pi_i$ to a path $\rho_i$ in $G$
and each variable $\nu_i$ to a tuple in $\{ \bot, \$, 1,\dots,n_2+1\}^k$ such that
\[
(G,\alpha) \vDash \psi.
\]
Let $G'$ be the graph $(V,E',\kappa)$  where $E'$ is the set 
\[
\{(i,a,i+1): 1 \leq i < n_2 \} \cup \{ (i,a,\$) : \text{ for some $1 \leq
  j \leq n_2$,
  $(i,\$)$ occurs in $\rho_j$ }\}
\]
That is, we delete the edges $(i,\$)$ that $\alpha$ ``does not use''. 
By definition of $E'$, the formula $\psi$ remains true in $G'$ under
the assignment $\alpha$. In particular, $\phi$ is true in $G'$. This
implies that $\exists x,y \; Q$ holds
in $G'$. 

Now, for all $1 \leq j \leq n_2$, there is at most one natural number
$i$ such that $(i,a,\$)$ occurs in $\rho_j$. This is simply because
there is no path going through edges $(i,a,\$)$ and $(i',a,\$)$ if $i \neq
i'$. This implies that the set 
\[
\{ (i,a,\$) : \text{ for some $1 \leq
  j \leq n_2$,
  $(i,\$)$ occurs in $\rho_j$ }\}
\]
contains at most $n_2$ edges. Since $G$ admits $n_2+1$ edges of the
form $(i,a,\$)$, there must be an edge $(i_0,a,\$)$ occurring in $G$, but
not in $G'$. This means that $G'$ is a graph of the form
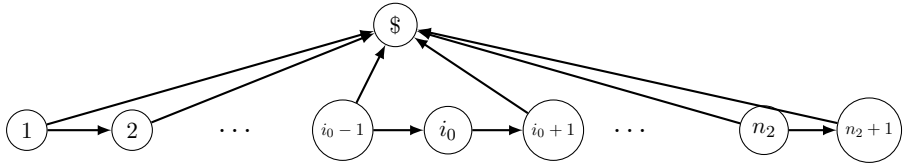
\begin{center}\begin{tikzpicture}[scale=0.7]
\node(s)  [draw, circle,scale=0.8] at (0,0) {$\$$};
\node(1)  [draw, circle,scale=0.8] at (-7,-2) {$1$};
\node(2)  [draw, circle,scale=0.8] at (-5,-2) {$2$};
\node(n)  [draw, circle,scale=0.8] at (7,-2) {$n_2$};
\node(n2)  [draw, circle,scale=0.6] at (9,-2) {$n_2+1$};
\node(i0)  [draw, circle,scale=0.8] at (1,-2) {$i_0$};
\node(i)  [draw, circle,scale=0.6] at (-1,-2) {$i_0-1$};
\node(i1)  [draw, circle,scale=0.6] at (3,-2) {$i_0+1$};
\node at (4.5,-2) {$\dots$};
\node at (-3,-2) {$\dots$};

\draw [->,>=latex, thick]
    (i1) -- (s);
\draw [->,>=latex, thick]
    (i) -- (s);
\draw [->,>=latex, thick]
    (i) -- (i0);
\draw [->,>=latex, thick]
    (i0) -- (i1);
\draw [->,>=latex, thick]
    (1) -- (s);
\draw [->,>=latex, thick]
    (2) -- (s)
;
\draw [->,>=latex, thick]
    (n2) -- (s)
;
\draw [->,>=latex, thick]
    (n) -- (s)
;
\draw [->,>=latex, thick]
    (1) -- (2)
;
\draw [->,>=latex, thick]
    (n) -- (n2)
;
\end{tikzpicture}
\end{center} 
In particular, $\exists x,y \; Q$ is not true in $G'$, which contradicts the fact
that 
$\phi$ is true in $G'$. 
\end{proof}

On the other hand, increasing the expressive power of RL$^+$ with some simple
forms of negation leads to intractability of query evaluation in data complexity: 

\begin{proposition} \label{prop:neg}
  There is a finite alphabet $\Sigma$ and REMs $e_1,e_2,e_3,e_4$ over
  $\Sigma$ and a single register $r$, such that {\sc Eval}{\em
    (RL,$\phi$)} is {\sc Pspace}-complete, where $\phi$ is either
  $\exists \pi \exists \lambda \neg (e_1(\pi,\bot,\lambda) \vee
  e_2(\pi,\bot,\bot))$ or $\exists \pi \forall \lambda \neg (e_3(\pi,\bot,\lambda) \vee
  e_4(\pi,\bot,\bot))$.
\end{proposition}

\begin{proof}
 By the proof of Theorem~\ref{theo:rl-dc} (and using the same
notation), we know that for every Turing machine $M$, there is a formula $\chi(\pi)$ such that for
all words $w$ of size $n$, there is a graph $G_w$ (of size polynomial in $n$) such that
\[
G_w \vDash \exists \pi \neg \chi (\pi) \quad \text{iff} \quad \text{
  there is an accepting run of $M$ over $w$ using at most $cn$ cells.}
\] 
Moreover, the formula $\chi(\pi)$ is a formula of the form 
\[
\bigvee \{ e_i (\pi,\bot,\bot) : i \in I \} \vee
\bigvee \{ \exists \lambda \; f_j (\pi,\bot,\lambda) : j \in J \},
\]
where $\{ e_i,f_j : i \in I,j\in J\}$ is a set of REMs that do not
contain any $\cup$. Since REMs are closed under union, $\chi(\pi)$ is
equivalent to a formula of the form
\[
\exists \lambda \; (e(\pi,\bot,\bot) \vee f(\pi,\bot,\lambda)).
\]
Hence, the formula $
\exists \pi \neg \chi (\pi) $ is equivalent to 
\[
\exists \pi \forall
\lambda \neg (e(\pi,\bot,\bot) \vee f(\pi,\bot,\lambda)).
\]
This proves that {\sc Eval}(RL,$\phi$) is {\sc
  Pspace}-complete, where $\phi = \exists \pi \forall
\lambda \neg (e(\pi,\bot,\bot) \vee f(\pi,\bot,\lambda))$.

Now we prove that {\sc Eval}(RL,$\phi'$) is {\sc
  Pspace}-complete. where $\phi'$ is a formula of the form $\exists
\pi \exists 
\lambda \neg (e'(\pi,\bot,\bot) \vee f'(\pi,\bot,\lambda))$ , for some
REMs $e'$ and $f'$. The intuition is as follows. The difference
between $\phi$ and $\phi'$ is that in $\phi$, we may choose the data
value that is in the register after checking that $f$ is
true. However, in $\phi'$, we must be able to store any value in the
register after checking that $f'$ is true. We will make two changes to
make this possible. 

First, we modify the graph $G_w$ in such a way that two arbitrary
nodes are always reachable. This can be easily achieved by adding 
an edge from the ``right-most node'' of the graph  $H$ to the
``left-most node'' of the graph $I_w$ (allowing to
encode the
initial configuration of a run). Second, we modify the REMs of $\phi$
in such a way that the label of a path satisfying those REMs, encodes an accepting
run and after reaching the final state, it goes through all the nodes
of $G_w$. Hence, once we checked that the run reaches the final state,
we can simply store any value in the register. We leave out the
details, as the intuition is pretty simple and the details a bit
tedious.

\end{proof}

In the case of basic navigational languages for graph databases, 
it is possible to increase the
expressive power -- without affecting the cost of evaluation -- by
extending formulas with a branching operator (in the
style of the class of {\em nested regular expressions} \cite{BLP}). 
The same idea can be applied in our scenario, by 
extending atomic REM formulas in
RL$^+$ with such a branching operator. The resulting language is
more expressive than RL$^+$ (in particular, this extension can express
query (Q)), yet remains tractable in data complexity. We formalize
this idea below.


The class of {\em nested} REMs (NREM) extends REM with a nesting 
operator $\langle \cdot \rangle$ defined as follows: 
If $e$ is an NREM then $\langle e \rangle$ is also an NREM. 
Intuitively, the formula $\langle e \rangle$ filters those nodes in a data graph 
that are the origin of a path that can be parsed according to $e$. Formally, 
if
$e$ is an NREM over $k$ registers and $G$ is a data graph, 
then $\sem{\langle e \rangle}_G$ consists of all tuples of the form $(u,\lambda,\rho = u,u,\lambda)$
such that $(u,\lambda,\rho',v,\lambda') \in \sem{e}_G$, for some node $v$ in $G$, path $\rho'$ in $G$, 
and $k$-tuple $\lambda'$ over $\D_\bot$. 

Let NRL$^+$ be the logic that is obtained from RL$^+$ by allowing atomic formulas of the form $e(\pi,\lambda,\lambda')$, 
for $e$ an NREM. Given a data graph $G$ and an assignment $\alpha$ for $\pi$, $\lambda$ and $\lambda'$ over $G$, 
we write as before $(G,\alpha) \models e(\pi,\lambda,\lambda')$ if and only if $\alpha(\pi)$ goes from $u$ to $v$ and 
$(u,\alpha(\lambda),\alpha(\pi),v,\alpha(\lambda')) \in \sem{e}_G$.  The semantics of NRL$^+$ is thus obtained from the semantics 
of these atomic formulas in the expected way. The following example
shows that query (Q) is expressible in NRL$^+$, and, therefore, that
NRL$^+$ increases the expressiveness of RL$^+$. 

\begin{example} Over graph databases, the query (Q) from the
introduction is expressible in NRL$^+$ using the following formula
over $\Sigma = \{a\}$ and register $r$: $$\phi \ = \ \exists \pi
\exists \lambda \, \big( \, (x,\pi,y) \, \wedge \, e(\pi,\lambda,\lambda) \,
\big),$$ where $e := (\langle e_1 \rangle \cdot a)^* \langle e_1
\rangle$, for $e_1 = a^*[r^=]$. Intuitively, $e_1$ checks in a path
whether its last node is precisely the node stored in register $r$,
and thus $e$ checks whether every node in a path can reach the node
stored in register $r$. Therefore, the formula $\phi$ defines the set
of pairs $(x,y)$ of nodes, such that there is a path $\pi$ that goes
from $x$ to $y$ and a register value $\lambda$ (i.e., a node $\lambda$) that
satisfies that every node in $\pi$ is connected to $\lambda$.  \boxtheorem
\end{example}

The extra expressive power of 
NRL$^+$ over RL$^+$ does not affect the data complexity of query evaluation: 

\begin{theorem} \label{theo:nrl-dc} 
Evaluation of NRL$^+$ formulas is in \nlogspace\ in terms of data complexity. 
\end{theorem} 
   
\begin{proof}
Let $G = (V,E,\kappa)$ be a data graph and $\phi$ an NRL$^+$
formula. Also, let $D = \{\kappa(v) \mid v \in V\}$. 
We assume without loss of generality that $\phi$ is Boolean,
that is, we study the complexity of deciding whether $G \models \phi$. 
In the case when $\phi$ is not Boolean, that is, when the input
consists of $G$ and an assignment $\alpha$ for $\phi$ over $G$, we
simply replace each free variable $\eta$ in $\phi$ by 
 $\alpha(\eta)$, and then use the evaluation
algorithm we describe below for the resulting formula.  

Assume without loss of generality that $\phi$ is of the form $\exists
\bar x \exists \bar \nu \exists \bar \pi \psi$, where $\bar x$ is a
tuple of node variables, $\bar \nu$ is a tuple of register assignment
variables, $\bar \pi$ is a tuple of path variables, and $\psi$ is
quantifier-free. Assume also that $\{e_1,\dots,e_m\}$ is the set of
NREMs mentioned in $\phi$, and that each such NREM is over
$\{r_1,\dots,r_k\}$.  The evaluation algorithm does the following: It
first guesses witnesses for the existentially quantified node and
register assignment variables. Assume that the guess for $\bar x$ is
$\bar v$, where $\bar v$ is a tuple of nodes of the same arity than
$\bar x$, and that the guess for $\bar \nu$ is $\bar \lambda$, where
$\bar \lambda$ is a tuple of elements in $(D \cup
\{\bot\})^k$ of the same arity than $\bar \nu$. 
Clearly, both $\bar v$ and $\bar \lambda$ can be
represented using only logarithmic space (since $\phi$, and therefore
$k$, is fixed).  

The algorithm then guesses a witness $\xi$ for each
existentially quantified path variable $\pi$ in $\bar \pi$. This
witness codifies all the information that we need to know about the
actual path $\rho$ that inteprets $\bar \pi$. In our case, $\xi$ is a
string that lists (in a precise order) the endpoints $u$ and $v$ of
$\rho$ and the tuples $(e,\lambda,\lambda')$, for $e \in
\{e_1,\dots,e_m\}$ and $\lambda,\lambda' \in \bar \lambda$, such that
$(u,\lambda,\rho,v,\lambda') \in \sem{e}_G$.
 Let $\bar \xi$ be the
witnesses for the tuple $\bar \pi$. It is not hard to see that $\bar
\xi$ can be represented using logarithmic space.

Finally, the algorithm checks that the guess $(\bar v,\bar
\lambda,\bar \xi)$ satisfies $\psi$. Since $\psi$ is a Boolean
combination of atomic formulas, we only have to explain how to do this
in logarithmic space for each atomic formula of the logic. Atomic formulas of
the form $x = y$, $\nu = \nu'$ and $\nu = \bar \bot$ are self-evident.
Atomic formulas of the form $\pi = \pi'$ only require checking whether
the witness of $\pi$ is equal to the witness of $\pi'$. Atomic
formulas of the form $(x,\pi,y)$ require checking in the witness $\xi$
of $\pi$ whether its endpoints correspond to the witnesses of $x$ and
$y$. Finally, formulas of the form $e(\pi,\nu,\nu')$ require checking
in $\xi$ whether $(u,\lambda,\rho,v,\lambda') \in \sem{e}_G$, where
$\lambda$ and $\lambda'$ are the witnesses of $\nu$ and $\nu'$,
respectively. Clearly, any of this can be done in logarithmic space. 

The only thing that remains to be done is checking that $\bar \xi$ is
consistent with $G$. That is, we need to check the following 
for each path variable $\pi$ that is mentioned in $\bar \pi$ whose
witness in $\bar \xi$ is $\xi$: There is a path $\rho$ in $G$
whose codification corresponds to $\xi$. In other words, if $\xi$
tells us that the endpoints of the path are $u$ and $v$, we need to
check in $G$ if there is a path $\rho$ from $u$ to $v$ such that, for
each $e \in \{e_1,\dots,e_m\}$ and $\lambda,\lambda' \in \bar
\lambda$, it is the case that $(u,\lambda,\rho,v,\lambda') \in
\sem{e}_G$ each time that $\xi$ tells us so. We explain next 
how this can be done in \nlogspace\ combining techniques from REM
and nested regular expressions evaluation.

The algorithm starts computing, for each expression of the form $\langle e \rangle$
that appears in any of the $e_i$'s, the set $U(e)$ 
of pairs of the form $(w,\lambda)$, for $w$ a node in $G$ and
$\lambda$ a $k$-tuple
 over $D \cup \{\bot\}$, such that $(w,\lambda,\rho' =
w,w,\lambda) \in \sem{\langle e \rangle}_G$. In order to do so it 
proceeds recursively depending on the {\em nesting depth} of the
expression $\langle e \rangle$, which is 1 if $e$ contains no nested
subexpression, and it is 1 more than the maximum
nesting depth of any subexpression of $e$ of the form $\langle e'
\rangle$ otherwise. The algorithm starts with those expressions $\langle e
\rangle$ of nesting depth one, i.e., when $e$ is an REM (no nesting). 
Using techniques from
\cite{LV} it is possible to compute in \nlogspace\ the set $U(e)$, for
each such expression $\langle e \rangle$. Then it continues with
expressions of nesting depth two. In such case, it uses the same
aforementioned techniques, but each time the procedure is asked to
check whether $(w,\lambda,\rho' =
w,w,\lambda) \in \sem{\langle e' \rangle}_G$, for a subexpression
$\langle e' \rangle$ of nesting depth one, it simply checks whether
$(w,\lambda) \in U(e')$. The process continues iteratively in this
way until all sets $U(e)$, for $\langle e \rangle$ a subexpression of
any of the $e_i$'s, are computed. 
Clearly, this iterative process can be performed in \nlogspace. 

Once the previous step is finished, the algorithm checks whether there
is a path $\rho$ from $u$ to $v$ such that, for each $e \in
\{e_1,\dots,e_m\}$ and $\lambda,\lambda' \in \bar \lambda$, it is the
case that $(u,\lambda,\rho,v,\lambda') \in \sem{e}_G$ each time that
$\xi$ tells us so. This can be done in \nlogspace\
applying the same techniques used in the previous paragraph and the
knowledge provided by the sets $U(e)$. The result follows from the
fact that \nlogspace\ functions are closed under composition. 
\end{proof}

From the proof of Theorem \ref{theo:nrl-dc} it also follows that
NRL$^+$ formulas can be evaluated in \pspace\ in combined complexity.

\section{Conclusions and Future Work}
\label{sec:conc}
We have proven that the data complexity of walk logic (WL) is nonelementary, which
rules out the practicality of the logic. We have proposed register
logic (RL),
which is an extension of regular expressions with memory. Our results in this 
paper suggest that register logic is capable of expressing natural queries 
about interactions between data and topology in data graphs, while still
preserving the elementary data complexity of query evaluation (\pspace). 
Finally, we showed how to make register logic more tractable in data
complexity (\nlogspace) through the logic NRL$^+$, while at the same time 
preserving some level of expressiveness of RL.

We leave open several problems for future work. One interesting question is
to study the expressive power of extensions of walk logic, in comparison to
RL and ECRPQ$^\neg$ from \cite{BLWW}. For example, we can consider
extensions with regularity tests (i.e. an atomic formula testing whether a 
path belongs to a regular language). Even in this simple case, the expressive
power of the resulting logic, compared to RL and ECRPQ$^\neg$, is already
not obvious. Secondly, we do not know whether NRL$^+$ is strictly more expressive than 
RL. 
Finally, we will also mention that expressibility of bipartiteness in WL
is still open (an open question from \cite{WL}). We also leave open whether
the query that a graph database is a connected graph with an even number of nodes
is expressible in WL.

\end{document}